\documentclass[10pt,conference,a4paper]{IEEEtran}
\IEEEoverridecommandlockouts
\usepackage{CJK}
\usepackage{cite}
\usepackage{float}
\usepackage{graphicx}
\usepackage{amsmath}
\usepackage{times}
\usepackage{graphicx}
\usepackage{latexsym}
\usepackage{graphicx}
\usepackage{bm}
\usepackage{amssymb}
\usepackage[center]{caption2}
\usepackage{stfloats}
\usepackage{cases}
\usepackage{array}
\usepackage{setspace}
\usepackage{fancyhdr}
\usepackage{color}
\usepackage{subfigure}
\usepackage{citesort}
\usepackage{multirow}
\usepackage{amsmath,amsthm}
\usepackage{cases}
\usepackage[noend]{algpseudocode}
\usepackage{algorithmicx,algorithm}
\usepackage{mathtools}

\theoremstyle{remark}

\def\bm#1{\mbox{\boldmath $#1$}}
\newtheorem{theorem}{Theorem}

\newtheorem{proposition}{Proposition}
\newtheorem{remark}{Remark}
\renewcommand{\baselinestretch}{0.97}

\renewcommand{\baselinestretch}{0.97}

\newcaptionstyle{mystyle2}{%
\captionlabel $.$ \, \captiontext \par} \captionstyle{mystyle2}

\begin{document}
\title{Programmable Metasurface Based Multicast Systems: Design and Analysis}

\author{\IEEEauthorblockN{
Xiaoling Hu,  \emph{Student Member}, \emph{IEEE},
Caijun Zhong, \emph{Senior Member}, \emph{IEEE},
Yongxu Zhu,  \emph{Member}, \emph{IEEE},\\
Xiaoming Chen,  \emph{Senior Member}, \emph{IEEE},
and Zhaoyang Zhang, \emph{Member}, \emph{IEEE}
\thanks{X. Hu, C. Zhong, X. Chen, and Z. Zhang are with the College of Information Science and Electronic Engineering, Zhejiang University, Hangzhou, China (email: caijunzhong@zju.edu.cn).}
\thanks{Y. Zhu is with the Wolfson School of Mechanical, Electrical and Manufacturing Engineering, Loughborough University, Leicestershire, LE11 3TU, UK (Email: y.zhu4@lboro.ac.uk).}
}}

\maketitle
\begin{abstract}
This paper considers a multi-antenna multicast system with programmable metasurface (PMS) based transmitter. Taking into account of the finite-resolution phase shifts of PMSs, a novel beam training approach is proposed, which achieves comparable performance as the exhaustive beam searching method but with much lower time overhead. Then, a closed-form expression for the achievable multicast rate is presented, which is valid for arbitrary system configurations. In addition, for certain asymptotic scenario, simple approximated expressions for the multicase rate are derived. Closed-form solutions are obtained for the optimal power allocation scheme, and it is shown that equal power allocation is optimal when the pilot power or the number of reflecting elements is sufficiently large. However, it is desirable to allocate more power to weaker users when there are a large number of RF chains. The analytical findings indicate that, with large pilot power, the multicast rate is determined by the weakest user. Also, increasing the number of radio frequency (RF) chains or reflecting elements can significantly improve the multicast rate, and as the phase shift number becomes larger, the multicast rate improves first and gradually converges to a limit. Moreover, increasing the number of users would significantly degrade the multicast rate, but this rate loss can be compensated by implementing a large number of reflecting elements.
\end{abstract}
\begin{IEEEkeywords}
Programmable metasurface, multicast systems, channel estimation.
\end{IEEEkeywords}

\section{Introduction}

By the year 2022, there will be 28.5 billion networked devices, and the overall mobile data traffic will reach up to 77 exabytes per month according to Cisco Visual Networking Index forecast\cite{dl1}. The tremendous growth in the number of communication devices calls for green and energy-efficient wireless solutions. To tackle this issue, the programmable metasurface (PMS), also known as intelligent reflecting surface (IRS), has recently been proposed as a promising solution due to its potential of both low power consumption and low deployment cost\cite{huang2018energy,qingqing2019towards,huang2019reconfigurable}.

Specifically,  a typical PMS is usually a uniform planar array composed of a large number of low-cost, passive, reflecting elements (e.g. printed dipoles), each of which can independently reflects the incident wireless signal with adjustable phase shift (controlled by an attached smart controller).
By adaptively tuning the phase shift of the reflecting elements, the propagation of the incident signal can be adjusted in a desirable way, thereby realizing smart and programmable wireless environment \cite{cui2014coding,liaskos2018new}.

%It is worth noting that compared with other related technologies, such as relay and phased array antennas, the PMS has the following advantages:
%\begin{itemize}
%\item
%
%First, compared with the AF relay which amplifies and regenerates signal actively,  the PMS only reflects the received signal in a passive way, and ideally does not need dedicated energy source, making it more power-saving.
%
%\item
%Secondly, by applying a pin-diode in each reflecting elements, the
%reflecting phases can be switched by simply turning the biasing direct-current (DC) voltage, acting as a programmable metasurface \cite{cui2014coding}.
%In contrast, the phased array antennas requires complex circuits to compensate for the distribution losses in the feeding network, causing high hardware cost.
%
%\item
%Thirdly, the PMS can be flexibly deployed, such as on walls or ceilings, due to its low profile and conformal geometry.
%\end{itemize}

Therefore, PMS-empowered wireless communications have attracted considerable research interests from both academia and industry. In general, the applications of PMSs can be divided into two catagories. One typical application is to use the PMS as a passive relay to assist in the communication from the transmitter to the receiver\cite{wu2019beamforming,chen2019intelligent,huang2019indoor}. Specifically, the PMS is deployed between the transmitter and receiver. Each PMS is connected with a controller which communicates with the transmitter via a separate wireless control link for coordination and exchanging channel state information (CSI) and smartly adjusts the phase shifts of reflecting elements. Such communication mode is especially useful when the direct link between  the transmitter and receiver is blocked\cite{huang2018achievable,bjornson2019intelligent,ye2019joint}. For example, assuming no line-of-sight communication is present, the work \cite{huang2018achievable} investigated a PMS-aided multiple input single output (MISO) communication systems, showing that the use of PMS increases the system throughput by at least $40\%$, without requiring any additional energy consumption. Also, there are  some works studying the utilization of PMSs in the presence of direct links \cite{wu2019beamforming,yan2019passive}.
However, using the PMS as a passive relay has two main disadvantages in practical systems.
\begin{itemize}
\item First, the PMS is far from the transmitter, making it difficult to obtain information (e.g.CSIs) from the transmitter, due to its passive architecture. To tackle this problem, a two-mode PMS model was proposed in  \cite{wu2019intelligent,subrt2012intelligent}, where the PMS is equipped with a controller that switches between receiving mode for CSIs and reflecting mode for data transmission.
However, the realization of receiving mode requires the deployment of receive radio frequency (RF) chains, leading to more hardware cost.
\item Secondly, as pointed out in \cite{bjornson2019demystifying}, instead of deploying the PMS  between  the transmitter and receiver, placing the PMS right at the transmitter or receiver will cause less power loss.
\end{itemize}

To overcome these drawbacks, another more practical application  of the PMS is to use the PMS as a component of the transmitter.\footnote{There are two main advantages of deploying PMS aided transmitter compared to having an active large intelligent surface (LIS) \cite{hu2018beyond}. First, the PMS aided transmitter can be easily realized by combining traditional horn antennas with  PMSs. Besides,  the PMS aided transmitter has the advantages of low cost and low power consumption, due to the passive architecture of the PMS.
However, the limited number of RF chains in the PMS aided transmitter makes it difficult to connect too many devices, while it has been shown in \cite{hu2018beyond} that a fair small LIS can connect quite a large number of devices.} Specifically, the PMS is deployed right at the transmitter, and each PMS cooperates with a RF chain. The signal transmitted from the RF chain is reflected by the PMS with little power loss, due to very short distance between the RF chains and the PMS. Moreover, the PMS controller is connected with the base station (BS), making it easier for the PMS to access the CSI information, thereby facilitating the joint design of phase shifts and digital beamformer.
Furthermore, experimental results have demonstrated that the PMS-based transmitter is feasible \cite{tang2019wireless,tang2019programmable}.
For instance, a PMS-based transmitter presented in \cite{tang2019wireless} has realized single carrier quadrature phase shift keying (QPSK) transmission over the air, achieving a data rate of $2.048$ Mbps, which is comparable to that achieved by the conventional method but with much lower hardware complexity. Later on, the work \cite{tang2019programmable} realized a PMS-based 8-phase shift-keying (8PSK) transmitter which can achieve a higher data rate of 6.144 Mbps over the air.

However, very few works have investigated the theoretical limits of communication systems with PMS-based transmitter\cite{tang2019wireless,tang2019programmable}. Also, the existing experiments all focus on the scenario with only a single RF chain. Motivated by these observations, in this paper, we propose a PMS-based transmitter including multiple RF chains for multicast communication systems, taking into account of finite phase shifts, and present a detailed analysis on the achievable system performance. To the best of our knowledge, this is the first attempt to provide theoretical analysis for communication systems with PMS-based transmitter. The main contributions of this paper are summarized as follows:
\begin{itemize}
\item A novel channel estimation scheme including phase shift beam training and equivalent channel estimation has been proposed. Simulation result shows that the proposed phase shift beam training algorithm achieves good performance but with much lower time overhead.
\item A closed-form expression is derived for the achievable rate of individual users, which enables efficient evaluation of the multicase rate, as well as reveals the impact of key system parameters on the user rate.
\item For some asymptotic scenarios, such as large pilot power, large number of RF chains, and large number of reflecting elements, closed-form solutions are derived for the optimal power control coefficients and the corresponding multicast rate.
\end{itemize}

The remainder of the paper is organized as follows. In Section \ref{s1}, we introduce the PMS-based multicast system, while in Section \ref{s2}, we propose a channel estimation scheme including phase shift beam training and equivalent channel estimation. Then, the achievable rate is derived in Section \ref{s3}, based on which we investigate the optimal power control coefficients and give a detailed analysis on the multicast rate in Section \ref{s4}. Numerical results and discussions are provided in Section \ref{s5}, and finally Section \ref{s6} concludes the paper.

Notation: Boldface lower case and upper case letters are used for column vectors and matrices, respectively. The superscripts ${\left(\right)}^{*}$, ${\left(\right)}^{T}$, ${\left(\right)}^{H}$, and ${\left(\right)}^{-1}$ stand for the conjugate, transpose, conjugate-transpose, and matrix inverse, respectively. Also, the Euclidean norm and absolute value are denoted by $\left\| \cdot \right\|$ and $\left|\cdot\right|$, respectively. In addition, $\mathbb{E}\left\{\cdot\right\}$ is the expectation operator, and $\text{tr}\left(\cdot\right)$ represents the trace.   And, $j$ of $e^{j \theta}$ denotes the imaginary unit.
Finally, $z \sim \mathcal{CN}(0,{\delta}^{2})$ denotes a circularly symmetric complex Gaussian random variable (RV) $z$ with zero mean and variance $\delta^2$, and $z \sim \mathcal{N}(0,{\delta}^{2})$ denotes a real valued Gaussian RV.

\section{System Model}\label{s1}
We consider a single-cell multicast system as illustrated in Fig.\ref{f8}, where  the BS  equipped  with a PMS-based transmitter  communicates with a  group  of $K$ single-antenna users.

The partially connected architecture is adopted, which is realized by aligning the beam of each directional horn antenna to the corresponding sub-PMS consisting of $L=\frac{N}{N_{\text{RF}}}$  non-overlapping reflecting elements, where $N_\text{RF}$ is the number of RF chains (antennas) and $N$ is the total number of reflecting elements.
% Namely, each of the $N_{\text{RF}}$ RF chains is connected  with $L=\frac{N}{N_{\text{RF}}}$ non-overlapping reflecting elements 
\footnote{ 
Please note, the distance between the BS and the PMS is related to the carrier wavelength. In general, a smaller carrier wavelength implies a shorter distance.}
\footnote{It is worth noting that the proposed PMS transmitter architecture is different from the hybrid analog and digital beamforming transceiver structure. First, the methods to realize the  partially connected architecture are different.  In the proposed architecture, the  partially connected architecture is realized by aligning the beam of a directional horn  antenna to the corresponding sub-metasurface, while in the hybrid architecture,  the  partially connected architecture is realized by connecting each RF chain  to a subarray via phase shifters.
Secondly, in the proposed architecture, phase shifts are realized by the passive metasurface, while in the hybrid architecture, the adjustment of signal phases is realized by phase shifters which in general require complex circuits.
Moreover, in a more general full-connected case, at each reflecting element, signals from different RF chains are first combined and then reflected with the same phase shift, while at each antenna of the full-connected hybrid architecture, the signals from different RF chains are first adjusted with different phase shifts by different phase shifters and then combined together.}

%The $N_{\text{RF}}$ PMSs cooperate with the $N_{\text{RF}}$  RF chains respectively. In particular, the signal transmitted from the $i$-th RF chain is reflected by the $L=\frac{N}{N_{\text{RF}}}$ reflecting elements of the $i$-th PMS, where $N$ is the total number of reflecting elements.

The $i$-th sub-PMS consists of $L$ reflecting elements corresponding to the $i$-th RF chain. Each element of the $i$-th sub-PMS behaves like a keyhole. During the uplink transmission period, the reflecting element combines all the received signals and re-scatters the combined signal to the $i$-th RF chains, while during the downlink period,  the reflecting element combines signal from the $i$-th RF chain and re-scatters the signal as if from a point source.

Since the PMS is close to the BS, the channel between them can be modeled by a line-of-sight (LOS) channel.
Specifically, the channel from the $i$-th antenna (RF chain) to the $i$-th sub-PMS is given by
$
{\bf g}_{\text{B2P},i}^T=\alpha_\text{B2P} {\bf a}^T_i,
$
where $\alpha_\text{B2P}$ denotes the path loss coefficient given by $
G \frac{A_{\text e}}{ 4 \pi d_\text{B2P}^2  }$  , where $G$ is the antenna gain,  $A_\text{e}$ is the effective area of each reflecting element perpendicularly to the direction of propagation, and $d_\text{B2P}$ is the distance from the BS to the PMS. ${\bf a}^T_i$ is the array response vector of the $i$-th sub-PMS, whose elements have unit amplitude.

Let ${\bf c}= {\beta[ e^{j\theta_1},..., e^{j\theta_n},..., e^{j\theta_N}]}^T$
 denote the phase shift beam, where $\theta_n \in \left[0,2\pi\right)$ and $\gamma \in [0,1]$ are phase shift and
amplitude  coefficient, respectively.  The amplitude coefficient is given by $\beta=\gamma\alpha_\text{B2P}$ with $\gamma$ depicting the energy reflection efficiency of the PMS, while
the impact of the array response vector ${\bf a}_i$ is reflected in the phase shifts of ${\bf c}$.

In practice, the reflecting elements are controlled by the digital to analog converters (DACs), hence have finite phase shifts due to limited DAC resolution. Without loss of generality, we use $\mathcal{Q}$ to denote the set of all possible values of $\theta_n$, which has a cardinality of $M_\text{ph}$. Similarly, the set of all possible phase shift beams are denoted by $\mathcal{C}$, which has a cardinality of $M={M_{\text{ph}}}^{N}$.

%To fully achieve the potential of the PMS-based transmitter, the acquisition of CSI is very important.
We assume block-fading channels, i.e., the channels remain the same during each coherence interval and vary independently between different coherence intervals. The entire communication process can be separated into two phases during each coherence interval, namely, channel estimation and multicasting transmission, which we elaborate in the ensuing sections.

\begin{figure}[!ht]
  \centering
  \includegraphics[scale=0.4]{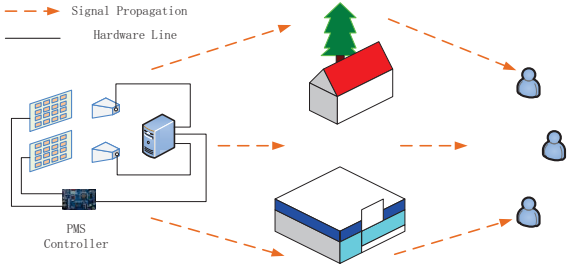}
   \caption{Model of the PMS-based transmitter aided multicast system with $N_{\text{RF}}=2, L=16, K=3$.}
  \label{f8}
\end{figure}
\section{Channel Estimation} \label{s2}

The proposed channel estimation scheme consists of two steps. In the first step, beam training is performed to acquire the optimal phase shift beam. In the second step, the equivalent channels are estimated.

\subsection{Beam Training}
Since the cardinality of phase shift beam set increases exponentially with the number of reflecting elements, the complexity of conventional exhaustive beam searching approach quickly becomes prohibitive. Responding to this, we propose a novel beam training algorithm.

Specifically, during the beam training phase, all $K$ users simultaneously transmit unmodulated frequency tones to the BS. For user $k$, the transmitted signal is denoted by $x_k=\sqrt{p_k}s$, where $p_k$ is the power and $s$ is the frequency tone of unit power.

For any $k$, we assume that
$
 \mathbb{E}\{ { \| \sqrt{p_k} {\bf g}_k \|}^2  \}=\varepsilon_r ,
$
with $\varepsilon_r$ being the average received power.  Also, ${\bf{g}}_{k}$ denotes the channel between the PMS and the $k$-th user and is defined as
$
{\bf{g}}_{k}=\sqrt{{\alpha}_{k}}{\bf{h}}_{k}
$,
where ${\alpha}_{k}$ models the large-scale fading, and ${\bf{h}}_{k}$ models the small-scale fading with elements being  independent and identically distributed (i.i.d) $\mathcal{CN}\left(0,1\right)$ RVs.
Furthermore, ${\alpha}_{k}$ is assumed to be constant and known as a priori.
After simplifying $ \mathbb{E}\{ { \| \sqrt{p_k} {\bf g}_k \|}^2 \}$, we have
$
 N\alpha_k p_k= \varepsilon_r.
$

%($\mathcal{C}_1=\mathcal{C}$)
The proposed beam training method works in a bisection manner, namely, at each stage, nearly half of the available beams will be eliminated. For instance, at the $i$-th stage, the BS chooses a pair of beams ${\bf{c}}_{(i,j)},j=1,2$, which have the weakest correlation from the current beam set $\mathcal{C}_i$. As such, the received signal after combining can be written as
\begin{align}\label{E2}
{\bf r}_{(i,j)}=\sum\limits_{k=1}^{K}{\bf C}_{(i,j)} {\bf g}_{k}x_k+{\bf n}_{(i,j)}, j=1,2,
\end{align}
where ${\bf C}_{(i,j)}  \in {\mathbb{C}}^{ N_{\text{RF}} \times N}$ is a block diagonal matrix defined by
\begin{equation} \label{E50} 
{\bf C}_{(i,j)}=\text{diag}\{ {\bf c}_{(i,j),1}^{T},\ldots, {\bf c}_{(i,j),n}^{T},\ldots,{\bf c}_{(i,j),N_{\text{RF}}}^{T}\}
\end{equation}
% \begin{equation}  \label{E50}     
% \left[                 
%   \begin{array}{ccccccc}    
%      {\bf c}_{(i,j),1}^{T} & {\bf 0} &\dots & \dots & \dots& \dots& {\bf 0}\\
%       {\bf 0} & {\bf c}_{(i,j),2}^{T}&{\bf 0} &\dots & \dots & \dots&   {\bf 0}\\
% \vdots& \vdots&\vdots & \vdots&\vdots&\vdots&\vdots \\
%      {\bf 0}& \dots & {\bf 0} &  {\bf c}_{(i,j),n}^{T} &{\bf 0}& \dots & {\bf 0}\\
% \vdots& \vdots&\vdots & \vdots&\vdots&\vdots&\vdots \\
%   {\bf 0}&\dots & \dots & \dots& \dots&{\bf 0} &{\bf c}_{(i,j),N_{\text{RF}}}^{T}
%   \end{array}
% \right]                 %ÓÒÀ¨ºÅ
% \end{equation}
with ${\bf c}^{T}_{(i,j),n}\triangleq {[{ c}_{(i,j),n}^1,...,{ c}_{(i,j),n}^k,...,{c}_{(i,j),n}^L]}^T
 \in {\mathbb{C}}^{L \times 1}$ being the phase shift vector of the $n$-th  sub-PMS.

The next step is to compare the received power ${\|{\bf r}_{(i,1)} \|}^2$ with ${\|{\bf r}_{(i,2)} \|}^2$. Let $j^\star= \underset {j=1,2}{\arg} \max {\|{\bf r}_{(i,j)} \|}^2$ and  $j^{-\star}= \underset {j=1,2}{\arg} \min {\|{\bf r}_{(i,j)} \|}^2$. It is intuitive that the optimal beam is more likely to have stronger correlation with ${\bf c}_{(i,j^{\star})}$. With this key observation, the number of training beams can be approximately halved by removing the beams which have weaker correlation with ${\bf c}_{(i,j^{\star})}$ . Specifically, the beam ${\bf c}  \in  \mathcal{C}_i$ satisfying ${\bf c}_{(i,j^{\star})}^H{\bf c} \le {\bf c}_{(i,j^{-\star})}^H{\bf c}$ will be removed, and the remaining beams makes up a new beam set ${\mathcal{C}}_{i+1}$. The process then continues until the cardinality of ${\mathcal{C}}_{i+1}$ becomes one. The pseudo-code of the proposed beam training method is summarized in Algorithm \ref{A1}. \footnote{It is worth highlighting that the proposed beam training method is substantially different from the beam training method used in the traditional hybrid architecture  \cite{xiao2017millimeter}. Specifically, in the proposed algorithm, the received signal power at the BS and the correlation between different beams are exploited to choose the best beam, while in the traditional hybrid architecture, the largest received SNR at the user and the beam-refinement protocol are utilized to choose the best beam. Moreover, the proposed algorithm does not require feedback from the users, which is necessary for the hybrid architecture.}

\begin{remark}
Since  our proposed beam training method works in a bisection manner, a much lower complexity of $O(\log_2(M_\text{ph}^N))$ can be achieved, compared with the complexity of exhaustive beam searching $O(M_\text{ph}^N)$.
\end{remark}

\begin{proposition} \label{p5}
When both $M_{\text{ph}}$ and $\varepsilon_r$ are sufficiently large, the ideal phase shift beam obtained by Algorithm \ref{A1} can be approximated by
\begin{align}
&{\bf{c}}_{\text{opt}}=
{[c_{\text{opt},1}^1,...,c_{\text{opt},1}^L,...,c_{\text{opt},N_{\text{RF}}}^1,...,c_{\text{opt},N_{\text{RF}}}^L]}^T,
\end{align}
where
$
c_{\text{opt},n}^{l}=
\beta \frac{h_{\text{sum},n}^l}{\left| {  h}_{\text{sum},n}^l \right|},\ \  {h}_{\text{sum},n}^l=\sum\limits_{k=1}^{K} { h}_{k,n}^l,
$
where ${h}_{k,n}^l$ denotes the small-fading coefficient between the $k$-th user and the $l$-th reflecting element of the $n$-th sub-PMS.
\end{proposition}

\begin{proof}
For notational convenience, we drop the subscript $(i,j)$ in (\ref{E2}) and we have
\begin{align} \label{E4}
{\bf r}&=\sum\limits_{k=1}^{K} {\bf C} {\bf g}_k x_k +{\bf n} =\sum\limits_{k=1}^{K} \sqrt{\alpha_k p_k} {\bf C} {\bf h}_k s+{\bf n}  \\
&\overset{(a)}{=} \sqrt{ \frac{\varepsilon_r}{N}} \sum\limits_{k=1}^{K}  {\bf C} {\bf h}_k s+{\bf n} \overset{(b)}{ \approx} \sqrt{\frac{\varepsilon_r}{N}} \sum\limits_{k=1}^{K}  {\bf C} {\bf h}_k s,\nonumber
\end{align}
where (a) is according to $
 N\alpha_k p_k= \varepsilon_r
$ and (b) follows the fact that $\frac{\varepsilon_r}{N}$ is sufficiently large.
Since the objective is to find the optimal phase shift beam ${\bf c}$ maximizing $\left\| {\bf r}\right\|$, we have the following equivalent optimization problem
\begin{align} \label{op1}
  \begin{array}{ll}
    \max\limits_{\left\{ {\bf c \in \mathcal{C}} \right\}}
&{\left\| {\bf r}\right\|}^2,
  \end{array}
\end{align}

Leveraging (\ref{E50}) and (\ref{E4}), we can express ${\left\| {\bf r}\right\|}^2$ as
\begin{align}
{\left\| {\bf r}\right\|}^2 =  \frac{\varepsilon_r}{N}{ \left\| \sum\limits_{k=1}^{K}{\bf C}{\bf h}_k \right\|}^2
=\frac{\varepsilon_r}{N} \sum\limits_{n=1}^{N_{\text{RF}}} {\left| {\bf c}_{n}^{T} \sum\limits_{k=1}^{K} {\bf h}_{k,n} \right|}^2,
\end{align}
where ${\bf h}_{k,n} $ denotes the channel vector between the $k$-th user and the $n$-th sub-PMS.

Based on the above equation, the optimization problem (\ref{op1}) can be rewritten as
\begin{align}
  \begin{array}{ll}
    \max\limits_{\left\{ {\bf c \in \mathcal{C}} \right\}}
&\sum\limits_{n=1}^{N_{\text{RF}}} { \left| {\bf c}_{n}^{T} \sum\limits_{k=1}^{K} {\bf h}_{k,n}  \right|}^2.
  \end{array}
\end{align}

Since the number of phase shifts, i.e., $M_{\text{ph}}$, is sufficiently large, we relax the elements of $\bf c$ to be complex numbers with continuous phases and fixed amplitudes, and obtain the following optimization problem:
 \begin{align}
  \begin{array}{ll}
     \max\limits_{{\bf c}}
&\sum\limits_{n=1}^{N_{\text{RF}}} {\left| {\bf c}_{n}^{T} \sum\limits_{k=1}^{K} {\bf h}_{k,n} \right|}^2,
  \end{array}
 \text{s.t.} & \begin{array}[t]{lll}
   \left| c_{n}^l\right|=\beta.
           \end{array}
\end{align}

Denote ${h}_{\text{sum},n}^l=\sum\limits_{k=1}^{K} { h}_{k,n}^l$. It is obvious that the phase of $c_{n}^{l}$ should equal to that of ${h}_{\text{sum},n}^l$, which completes the proof.
\end{proof}

 \begin{algorithm}[!t]         
\caption{Beam training algorithm}             
\label{A1}
\begin{algorithmic}
\State {\bf Initialize}:  stage number $i=0$, the training beam set of the first stage ${\mathcal{C}}_{1}=\mathcal{C}$.

\Repeat

 \State{  Set $i=i+1$.}

\State {Find $({\bf c}_{(i,1)},{\bf c}_{(i,2)})= \underset{({\bf v}_1, {\bf v}_2), {\bf v} _1,{\bf v} _2 \in {\mathcal{C}}_{i}}{\arg}  \min {\bf v} _1^H {\bf v} _2$.}

\State {The BS trains ${\bf c}_{(i,1)}$ and ${\bf c}_{(i,2)}$ respectively, and compares the received signal power  $ {\left\|{\bf r}_{(i,j)} \right\|}^2, j=1,2$.}

 \State {Let $j^\star= \underset {j=1,2}{\arg} \max {\left\|{\bf r}_{(i,j)} \right\|}^2$ and  $j^{-\star}= \underset {j=1,2}{\arg} \min {\left\|{\bf r}_{(i,j)} \right\|}^2$.}

  \State{Update the training beam set ${\mathcal{C}}_{i+1}=\left\{ {\bf c}\ |  \ {\bf c}_{(i,j^{\star})}^H{\bf c}< {\bf c}_{(i,j^{-\star})}^H{\bf c}, {\bf c} \in {\mathcal{C}}_{i}  \right\}$.
  }

\Until $\left|\mathcal{C}_{i} \right|=1$

\State output:${\bf c}_{(i,j^{\star})}$

\end{algorithmic}
\end{algorithm}

\subsection{Equivalent Channel Estimation}

Denote $\bar{\bf{h}}_k \triangleq {\bf{C}} {\bf{h}}_k  \in { \mathbb{C}}^{N_{\text{RF}} \times 1}$ and define the equivalent channel between the BS and the $k$-th user as $\bar{\bf{g}}_k= \sqrt{\alpha_k}\bar{\bf{h}}_k  \in { \mathbb{C}}^{N_{\text{RF}} \times 1}$. Note that ${\bf{C}}$ is the  phase shift matrix corresponding to the optimal phase shift beam obtained in the beam training phase.

Then we estimate the equivalent channel through uplink training, where all $K$ users simultaneously transmit orthogonal pilot sequences to the BS. Let ${\tau}_{c}$ be the length of the coherence interval (in symbols), and ${\tau}_{p}$ be the uplink training duration (in symbols) per coherence interval such that ${\tau}_{p}<{\tau}_{c}$. Denote the pilot sequence used by the $k$-th user, $k=1,2,...,K$, by $\sqrt{{\tau}_{p}}{\bm{\varphi}}_{k} \in {\mathbb{C}}^{{\tau}_{p}\times 1} $, where ${\rVert {\bm{\varphi}}_{k} \rVert}^{2}=1 $. To ensure the orthogonality of the pilot sequences, i.e. ${\bm \varphi}_i^H{\bm \varphi}_j=0, i \ne j$, it is required that $\tau_p \ge K$. Furthermore, we assume $\tau_p = K$.

Then, the ${ N_{\text{RF}} \times {\tau}_{p}}$ received pilot matrix at the BS can be expressed as
\begin{align}
{\bf{Y}}_{p}=\sqrt{{\tau}_{p} {\rho}_{p}} \sum\limits_{k=1}^{K} \bar{\bf{g}}_{k} {\bm{\varphi}}_{k}^{H}+{\bf{W}}_{p},
\end{align}
where ${\rho}_{p}$ is the normalized signal to noise ratio (SNR) of each pilot symbol, ${\bf{W}}_{p} \in {\mathbb{C}}^{N \times {\tau}_{p}}$ is the additive white Gaussian noise (AWGN) matrix, whose elements are i.i.d. $\mathcal{CN}(0,1)$ RVs.

To estimate $\bar{{\bf{g}}}_k$, we first multiply ${\bf{Y}}_{p}$ by ${\bm{\varphi}}_{k}$, which gives
\begin{align} \label{E5}
{\bf{y}}_{p,k}={\bf{Y}}_{p} {\bm{\varphi}}_{k}=\sqrt{{\tau}_{p} {\rho}_{p}}\bar{\bf{g}}_{k}+{\bf n}_{p,k},
\end{align}
where ${\bf n}_{p,k} = {\bf{W}}_{p}{\bm{\varphi}}_{k}$. The BS then adopts the minimum mean-square (MMSE) method to estimate the equivalent channel, as such, the equivalent channel $\bar{\bf{g}}_{k} $ can be decomposed as
\begin{align}
\bar{\bf{g}}_{k}={\hat{\bar{\bf{g}}}}_{k}+{\bf{e}}_{k},
\end{align}
where ${\hat{\bar{\bf{g}}}}_{k}$ is the estimation of $\bar{\bf{g}}_{k}$, ${\bf{e}}_{k}$ is the estimation error.

To obtain the distribution of the estimated equivalent channel, we first give an important proposition corresponding to the distribution of the equivalent channel.
 \begin{proposition} \label{p1}
With finite number of phase shifts, the elements of ${\bar{\bf h}}_k={\bf C}{\bf h}_k$ can be modeled as  
 i.i.d. random variables $\mathcal{CN} \left(u,\delta^2\right)$ with
\begin{align}
u& = \frac{L \beta}{2} \sqrt{\frac{\pi}{K}}  \frac{M_{\text{ph}}}{\pi} \sin\left(\frac{\pi}{M_{\text{ph}}}\right), \label{E11} \\
\delta^2 &=L \beta^2 \left\{1-\frac{\pi}{4 K}   {\left( \frac{M_{\text{ph}}}{\pi} \sin\left(\frac{\pi}{M_{\text{ph}}}\right)\right)}^2  \right\}.\label{E12}
\end{align}

\end{proposition}

\begin{proof}
See Appendix \ref{A2}.
\end{proof}

\begin{remark}
From Proposition $\ref{p1}$, we can see that the deployment of the PMS can enhance the equivalent channel compared to the case without the PMS. Specifically,  the strength of the channel without the PMS is $\alpha_k$, while the strength with the PMS is given by $\alpha_k(u^2+\delta^2)$, indicating that an asymptotic gain in the order of $\mathcal{O}\left( L^2\right)$ can be achieved. 
This is  because the PMS not only achieves the phase shift beamforming gain of order $L$ but also captures an inherent aperture gain of order $L$ by collecting more signal power.
\end{remark}

 Based on  Proposition \ref{p1} and the MMSE estimation property, ${\bf{e}}_{k}$ and  ${\hat{\bar{\bf{g}}}}_{k}$ are complex Gaussian distributed, and they are independent of each other. Then, we have the following proposition:
\begin{proposition}\label{p2}
The elements of ${\hat{\bar{\bf{g}}}}_{k}$and ${\bf{e}}_{k}$ are Gaussian RVs with the distributions $\mathcal{CN}\left( u_{p,k}, \delta_{p,k}^{2}\right)$ and
$\mathcal{CN}\left( 0, \delta_{e,k}^{2}\right)$ respectively, where
$
u_{p,k}=\sqrt{\alpha_k}u,
{\delta}_{p,k}^{2}= \frac{\tau_p \rho_p \alpha_k^2 \delta^4}{1+ \tau_p \rho_p \alpha_k \delta^2}, 
{\delta}_{e,k}^{2}= \frac{\alpha_k \delta^2  }{ 1 +\tau_p \rho_p \alpha_k  \delta^2}.
$
\end{proposition}

\begin{proof}
See Appendix \ref{A3}.
\end{proof}

\section{Achievable Rate Analysis} \label{s3}
During the multicasting phase, the BS utilizes the estimated equivalent CSI to precode the signals. To keep the processing simple, the BS adopts the  transmit matched
filter (MF) ${\bf W}={\hat{\bar{\bf G}}}^*$, then the received signal at all users is given by
\begin{align}
{\bf y}=\sqrt{\rho} { \bar{\bf{G}} }^T {\bf W} \sqrt{{\bf P}} {\bf s} +{\bf n},
\end{align}
where $\bar{\bf{G}}=[{\bf{g}}_1,...{\bf{g}}_k,...,{\bf{g}}_K]$, $\rho$ is the total average transmit power (normalized by the noise power), ${\bf P}=\text{diag}\{ \frac{\eta_1}{E \left\{ {\left\| {\bf w}_1\right\|}^2 \right\}},..., \frac{\eta_i}{E \left\{ {\left\| {\bf w}_i\right\|}^2 \right\}},..., \frac{\eta_K}{E \left\{ {\left\| {\bf w}_K\right\|}^2 \right\}} \}$ is the power control matrix
with the power control coefficient $\eta_k$, ${\bf s}={[s,...,s,...,s]}^T$ is the the data symbol vector satisfying $\mathbb{E} \{{|s|}^{2}\}=1$, and ${\bf n} \sim \mathcal{CN}\left( {\bf 0}, {\bf I}_K\right)$ denotes the noise.

 Noticing that ${ \bar{\bf{G}} }={ \hat{\bar{\bf{G}}} }+{\bf E}$, the above equation can be rewritten as
 \begin{align}
 {\bf y}= \sqrt{\rho} {\hat{\bar{\bf G}}}^T {\hat{\bar{\bf G}}}^* \sqrt{{\bf P}} {\bf s}  + \sqrt{\rho} {\bf E}^T  {\hat{\bar{\bf G}}}^* \sqrt{{\bf P}} {\bf s}  +{\bf n}.
 \end{align}

Then, the received signal at the $k$-th user is given by
\begin{align}
& y_k={\hat{\bar{\bf g}}}_k^T {\hat{\bar{\bf G}}}^* \sqrt{{\bf P}} {\bf s}  +{\bf e}_k^T  {\hat{\bar{\bf G}}}^* \sqrt{{\bf P}} {\bf s}  +n_k \\
 & \overset{(a)}{=} \sqrt{\rho} \left( {\hat{\bar{\bf g}}}_k^T+{\bf e}_k^T\right) \sum\limits_{i=1}^{K} \sqrt{\frac{\eta_i}{u_{p,i}^{2}+\delta_{p,i}^{2}}}  {\hat{\bar{\bf g}}}_i^* s
 +n_k, \nonumber
\end{align}
where (a) follows the fact that $\mathbb{E} \left\{ {\left\| {\bf w}_i\right\|}^2 \right\}= u_{p,i}^{2}+\delta_{p,i}^{2}$.

Next,  without loss of generality, let us focus on the achievable rate of the $k$-th user. We consider the realistic case where the $k$-th user does not have access to the instantaneous CSI of the effective channel gain. Instead, the detection of desired signal $s$ is based on the statistical CSI.  As such, we can rewrite $y_{k}$  as
\begin{align}
{y}_{k}
&=\underbrace{  \sqrt{\rho} \mathbb{E} \left\{  \left( {\hat{\bar{\bf g}}}_k^T+{\bf e}_k^T\right) \sum\limits_{i=1}^{K}  \sqrt{\frac{\eta_i}{u_{p,i}^{2}+\delta_{p,i}^{2}}}  {\hat{\bar{\bf g}}}_i^*   \right\} s }_{\operatorname{desired \ signal}}
 +\underbrace{{n}_{k}^{\operatorname{eff}}}_{\operatorname{effective \ noise}},
\end{align}
where
\begin{align}
{n}_{k}^{\operatorname{eff}}&=\underbrace{n_k}_{\text{noise}}+
 \sqrt{\rho} \left( {\hat{\bar{\bf  g}}}_k^T+{\bf e}_k^T\right) \sum\limits_{i=1}^{K} \sqrt{\frac{\eta_i}{u_{p,i}^{2}+\delta_{p,i}^{2}}}  {\hat{\bar{\bf g}}}_i^* s \\
&-\sqrt{\rho} \mathbb{E} \left\{  \left( {\hat{\bar{\bf g}}}_k^T+{\bf e}_k^T\right) \sum\limits_{i=1}^{K} \sqrt{\frac{\eta_i}{u_{p,i}^{2}+\delta_{p,i}^{2}}}  {\hat{\bar{\bf g}}}_i^*  \right\} s . \nonumber
\end{align}

Capitalizing on the results in \cite{marzetta2016fundamentals}, the achievable rate of the $k$-th user can be expressed as\footnote{ It is worth noting that this expression is derived under the assumption of the transmit MF and the realistic case where the users have no access to the instantaneous CSI of the effective channel gain.}
\begin{align}
{R}_{k}= {\log}_{2} \left(1+\frac{{|{A}_{k}|}^{2}}{ B_k  +1 }\right),
\end{align}
with
\begin{align}
&{A}_{k} \triangleq  \sqrt{\rho} \mathbb{E} \left\{  \left( {\hat{\bar{\bf g}}}_k^T+{\bf e}_k^T\right) \sum\limits_{i=1}^{K}  \sqrt{\frac{\eta_i}{u_{p,i}^{2}+\delta_{p,i}^{2}}}  {\hat{\bar{\bf g}}}_i^*  \right\} , \\
&{B}_{k} \triangleq \rho  \mathbb{E}  \left\{{\left| \left( {\hat{\bar{\bf g}}}_k^T+{\bf e}_k^T\right) \sum\limits_{i=1}^{K}  \sqrt{\frac{\eta_i}{u_{p,i}^{2}+\delta_{p,i}^{2}}}  {\hat{\bar{\bf g}}}_i^* \right|}^2 \right\}-{\left| A_k\right|}^2,
\end{align}
being the desired signal power and leakage power, respectively.

Then, we have the following important result:
\begin{theorem} \label{T1}
The achievable rate of the $k$-th user is given by
(\ref{E10}) on the top of the next page.
\newcounter{mytempeqncnt}
\begin{figure*}[!t]
\normalsize
\setcounter{mytempeqncnt}{\value{equation}}
 \begin{align} \label{E10}
R_k=\log_2 \left(1+\frac{  \rho  N_{\text{RF}}^2  {\left\{\sum\limits_{i=1}^{K} \sqrt{\frac{\eta_i}{u_{p,i}^{2}+\delta_{p,i}^{2}}} u_{p,k} u_{p,i}+\sqrt{\frac{\eta_k}{u_{p,k}^{2}+\delta_{p,k}^{2}}} \delta_{p,k}^{2} \right\}}^2 }    {1+\alpha_k \rho \delta^2 N_{\text{RF}} {\left( \sum\limits_{i=1}^{K} \sqrt{\frac{\eta_i u_{p,i}^2}{u_{p,i}^{2}+\delta_{p,i}^{2}}} \right)}^2
+\alpha_k\rho  \left(u^2+\delta^2\right)N_{\text{RF}} \sum\limits_{i=1}^{K}  \frac{\eta_i \delta_{p,i}^2}{u_{p,i}^{2}+\delta_{p,i}^{2}}} \right).
\end{align}
\vspace*{-0.8cm}
\end{figure*}
\end{theorem}
\begin{proof}
Refer to Appendix \ref{A4}.
\end{proof}
%Based on (\ref{section:Achievable Downlink Rate Analysis:R_k_finally}), we provide some insights into the performance of cell-free massive MIMO systerms with ADCs.

Theorem \ref{T1} presents a closed-form expression for the achievable rate which reveals the impact of key system parameters, such as the number of phase shifts, reflecting elements, RF chains and users, as well as the impact of imperfect channel estimation on the achievable rate. For instance, $R_k$ is an increasing function with respect to $N_{\text{RF}}$. Besides, it can be seen that  the desired signal power decreases with the equivalent channel estimation error, indicating that we can improve the channel estimation accuracy, for example by increasing the pilot power.
%To gain further insights, in the next section, we now consider the following scenario

After deriving the individual rate for any user $k \in \left\{1,2,...,K \right\}$, the multicast rate $R$ can be obtained as
\begin{align}
R= \min\limits_{k=1,2,...,K } {R}_{k}.
\end{align}

\section{Power Control}   \label{s4}

To maximize the multicast rate, we formulate the following power control problem:
\begin{align} \label{E24}
  \begin{array}{ll}
    \max\limits_{\left\{ {\eta}_{k} \right\}}
& \min\limits_{k=1,...,K} {R}_{k},  \\
 \operatorname{s.t.} & \begin{array}[t]{lll}
             \sum\limits_{k=1}^{K} {\eta}_{k}  = 1 \\
             {\eta}_{k} \geq 0, k=1,...,K.
           \end{array}
  \end{array}
\end{align}

In the general setting, the above optimization problem is a non-convex problem, hence is difficult to solve. Responding to this, we consider some asymptotic regime, where closed-form solutions can be derived.

\subsection{Large pilot power }
We first consider the scenario where the pilot power is sufficiently large, and we have the following important result:
\begin{theorem} \label{t2}
As $\rho_p \rightarrow \infty $, the optimal  power control coefficients
are
$
\eta_k=\frac{1}{K}, k=1,...,K,
$
 and the  corresponding multicast rate is given by
\begin{align}
& R=  \\ 
&\log_2\! \left(\!1\!+\!\frac{K N_{\text{RF}} L^3 \tilde{u}_0^4 }{K L^2 \tilde{u}_0^2 \!\left(\!1\!-\!\tilde{u}_0^2\!\right)\!
 \!+\!\left(\!L \tilde{u}_0^2\!+\!1-\!\tilde{u}_0^2 \!\right)\!\left(\! \frac{1}{\beta^2 N_{\text{RF}} \rho \alpha} \!+\!1\!-\!\tilde{u}_0^2\!\right)  } \!\right),\nonumber
\end{align}
where $ \alpha \triangleq \underset{k=1,...,K}{\min} \alpha_k$,  $\tilde u_0 \triangleq \frac{u_0}{\beta}=  \sqrt{\frac{\pi}{4K}}  \frac{M_{\text{ph}}}{\pi} \sin\left(\frac{\pi}{M_{\text{ph}}}\right)$.
\end{theorem}

\begin{proof}
Refer to Appendix \ref{A5}.
\end{proof}

Theorem \ref{t2} shows that, with large pilot power, the multicast rate is an increasing function with respect to $L$. This is because the equivalent channel can be enhanced by increasing the number of reflecting elements. Also, as the amplitude
reflection coefficient increases, the achievable rate becomes larger, due to the fact that larger amplitude reflection coefficient implies less power loss when the transmit signal is reflected by the PMS.
In addition, the multicast rate is a decreasing function with respect to $K$. This is reasonable because with fewer users, highly directional beams
can be obtained. Furthermore, the multicast rate is constrained by the large-fading coefficient of the weakest user, but this negative effect of the weakest user can be compensated by increasing the number of reflecting elements or RF chains.

%Corollary  \ref{c1} shows that, with perfect equivalent channel estimation, the effect of noise vanishes, and the Rician $K$-factor becomes irrelevant. Also, it can be seen that the power of the desired signal is a decreasing function with respect to $\kappa$. This is reasonable because a larger $\kappa$ means that the user is more likely to fall into the coverage of side lobes, leading to a severe power loss. Besides, we observe that the achievable rate is greatly influenced by the power allocation coefficients and spatial direction distances among different clusters, which implies that it is possible to improve the achievable rate by proper power control and cluster grouping design.

\subsection{A large number of RF chains}
%With a large number of reflecting elements, i.e. $N=N_{\text{RF}} L\to \infty$, we consider the following two cases.
\begin{theorem}\label{t3}
When $L$ is fixed while $N_{\text{RF}} \rightarrow \infty$, the optimal   power control coefficients
are
\begin{align} \label{E20}
\eta_k=\frac{\alpha_k \left(u_{p,k}^{2} +\delta_{p,k}^2\right) }{\phi  \delta_{p,k}^4},k=1,...,K,
\end{align}
where $\phi=\sum\limits_{k=1}^{K} \frac{ \alpha_k  \left(u_{p,k}^{2} +\delta_{p,k}^2 \right)}{ \delta_{p,k}^4} $.
And the corresponding   multicast rate is given by
 \begin{align}
&R=\\
&\log_2 \left(1\!+\!\frac{N_{\text{RF}} {\left\{ \frac{u^2}{\delta_2} \sum\limits_{k=1}^{K}\frac{1\!+\!\tau_p \rho_p \alpha_k \delta^2}{\tau_p \rho_p \alpha_k \delta^2}  \!+\!1\right\}}^2 }
{\frac{u^2}{\delta^2}{\left( \sum\limits_{k=1}^{K}\frac{1\!+\!\tau_p \rho_p \alpha_k \delta^2}{\tau_p \rho_p \alpha_k \delta^2}\right)}^2
\!+ \!\left( \frac{u^2}{\delta^2}\!+\!1 \right) \sum\limits_{k=1}^{K}\frac{1\!+\!\tau_p \rho_p \alpha_k \delta^2}{\tau_p \rho_p \alpha_k \delta^2}} \right).\nonumber
\end{align}
\end{theorem}
\begin{proof}
Refer to Appendix \ref{A6}.
\end{proof}

Theorem  \ref{t3} implies that with a large number of RF chains, the effect of noise vanishes, and the multicast rate is determined by the channel conditions of all users.
Moreover, the maximum signal to interference plus noise ratio (SINR) is proportional to $N_{\text{RF}}$, indicating that increasing the number of RF chains can significantly improve the multicast rate.
Besides, increasing the pilot power can improve the multicast rate, due to more accurate channel estimation.

\begin{proposition}\label{p6}
The power control coefficients $\eta_k$ is a decreasing function with respect to $\alpha_k$, indicating that more power should be allocated to users with poor channel conditions.
\end{proposition}
\begin{proof}
Utilizing the results given by Proposition \ref{p2},  the optimal power control coefficient can be rewritten as
$
\eta_k=\frac{u^2+\tilde{\delta}_{p,k}^2 }{\phi\tilde{\delta}_{p,k}^4 },
$
where $\tilde{\delta}_{p,k}^2 \triangleq \frac{\tau_p \rho_p \alpha_k \delta^4}{1+\tau_p \rho_p \alpha_k \delta^2}$.
Let $\tilde{\eta}_k \triangleq \phi \tilde{\eta}_k$. Due to $\frac{\partial \tilde{\eta}_k}{\partial \tilde{\delta}_{p,k} } <0$ and $\frac{ \partial \tilde{\delta}_{p,k} }{\partial \alpha_k} >0$,  we have
$
\frac{\partial \tilde{\eta}_k}{\partial \alpha_k} =\frac{\partial \tilde{\eta}_k}{\partial \tilde{\delta}_{p,k} }
\frac{ \partial \tilde{\delta}_{p,k} }{\partial \alpha_k} <0,
$
indicating that $\tilde{\eta}_k$ is  a decreasing function with respect to $\alpha_k$.
To this end, noticing that $\eta_k=1/({1+ \frac{1}{\tilde{\eta}_k}\sum\limits_{i \ne k}^{K}  \tilde{\eta}_i  })$ increases with $\tilde{\eta}_k$, we complete the proof.
\end{proof}

\subsection{A large number of reflecting elements}

\begin{theorem}\label{t4}
When $N_{\text{RF}}$ is fixed while $L \rightarrow \infty$, the optimal   power control coefficients
are
\begin{align} \label{E20}
\eta_k=\frac{1 }{K},k=1,...,K,
\end{align}
and the  multicast rate is given by
 \begin{align} \label{E14}
R =\log_2 \left\{1+\frac{ \pi N_{\text{RF}} L  {\left( \frac{M_{\text{ph}}}{\pi} \sin\frac{\pi}{M_{\text{ph}}}\right) }^2}{\left\{ 4-\frac{\pi}{K}{\left( \frac{M_{\text{ph}}}{\pi} \sin \frac{\pi}{M_{\text{ph}}}\right)}^2\right\} \left( K+1\right)}\right\}.
\end{align}
\end{theorem}
\begin{proof}
Refer to Appendix \ref{A7}.
\end{proof}

Theorem \ref{t4} shows that with a large $L$, the effect of noise as well as the equivalent channel estimation error vanishes. The reason is that a large number of reflecting elements can significantly enhance the equivalent channel. Also, as the number of reflecting elements increases, the amplitude reflection coefficient  becomes irrelevant, indicating that increasing the number of reflecting elements can compensate for the power loss caused by PMS reflection. In addition, the SINR is proportional to $N=N_{\text{RF}} L$, which implies that the multicast rate can be greatly improved by increasing the number of reflecting elements.

% Moreover, as the number of reflecting elements becomes larger, the achievable rates of all users approach the same value, which is determined by the number of RF chains, the number of reflecting elements, power allocation, user number and phase number.

\subsubsection{The impact of phase shift number}
\begin{proposition} \label{p3}
With large $L$, the  multicast rate is an increasing function with respect to the phase shift number $M_{\text{ph}}$. Furthermore, when the phase shift number is sufficiently large, the multicast rate is given by
\begin{align}
R =\log_2 \left\{1+\frac{\pi N_{\text{RF}} L   K }{\left( 4K-\pi\right) \left\{ K+1 \right\}} \right\}.
\end{align}

\begin{proof}
Starting from $R$ given in Theorem \ref{t4}, we can see that $ \frac{M_{\text{ph}}}{\pi} \sin\left(\frac{\pi}{M_{\text{ph}}} \right)$ and $R$ are increasing functions with respect to $M_{\text{ph}}$ and $ \frac{M_{\text{ph}}}{\pi} \sin\left(\frac{\pi}{M_{\text{ph}}} \right)$, respectively.  Thus, $R$ increases with $M_{\text{ph}}$.
Noticing that $\underset{M_{\text{ph}} \to \infty  }{\lim}  \frac{M_{\text{ph}}}{\pi} \sin\frac{\pi}{M_{\text{ph}}}=1$, we can obtain the desired result.
\end{proof}

Proposition \ref{p3} is rather intuitive since highly accurate beam can be obtained with high-resolution phase shifts. Moreover, as the phase shift number becomes sufficiently large, the multicast rate becomes independent of  $M_{\text{ph}}$ and gradually converges to a limit, indicating that the gain of using high-resolution phase shift diminishes gradually.

\end{proposition}

\subsubsection{The impact of user number}
\begin{proposition} \label{p4}
With a large number of reflecting elements, the multicast rate is a decreasing function with respect to the user number. Furthermore, with a large number of users, the multicast rate is given by
\begin{align}
R=\log_2 \left\{1+\frac{ \pi N_{\text{RF}} L  {\left( \frac{M_{\text{ph}}}{\pi} \sin\frac{\pi}{M_{\text{ph}}}\right) }^2}{4K}\right\}.
\end{align}
\end{proposition}

\begin{proof}
Starting from Theorem \ref{t4}, we can easily obtain the desired result.
\end{proof}

From Proposition \ref{p4}, it can be seen that with massive users, the  SINR is inversely proportional to the user number, which implies that increasing the number of users can severely degrade the multicast rate. To compensate this rate loss, it is desired to employ a large number of reflecting elements.

\subsubsection{Increasing $N_{\text{RF}}$ V.S. Increasing $L$ }
%\begin{remark}

 Although a higher rate can be achieved by increasing either $N_{\text{RF}}$ or  $L$ , it is better to increase the number of reflecting elements rather than the number of RF chains, because the power consumption and hardware cost of PMS are much lower than that of RF chains. In addition, with massive reflecting elements, the negative effects of noise, estimation error as well as the amplitude reflection coefficient can be effectively compensated, while with a large number of RF chains, only the effects of noise can be mitigated.
%\end{remark}

%From Proposition \ref{p4}, it can be seen that the optimal power control scheme is to equally allocate power to each user. In addition, with optimal power control, the achievable rate is a decreasing function with respect to $K$. This is reasonable because with fewer users, highly directional beams can be obtained.

\section{numerical  results} \label{s5}
In this section, we provide numerical results to illustrate the performance of the PMS-based multicast system, as well as to verify the performance of the proposed channel estimation scheme. The considered system is assumed to operate at the frequency of $f_c=4.25$ GHz with the bandwidth of 180 kHz,\footnote{ In practice, the metasurface can only handle a limited bandwidth, because the same phase shifts must be applied in the entire band. How to design the metasurface operating in a wider frequency band remains to be studied. } and the coherence time is $\sqrt{\frac{9}{16 \pi f_m^2}}$ with the maximum Doppler shift given by $f_m=1$ Hz.
The noise spectral power density is $-169$ dBm/Hz.
The channel from the transmitter to the user is modeled as Rayleigh fading.
The large-scale fading coefficient is given by $\alpha=L_{\text{d}}^{-\lambda}$, where $\lambda=3$ is the path loss exponent, and $L_{\text{d}}$ is the transmission distance.
The gain of each horn antenna is $20$ dBi.
The PMS deployed $1$m away from the BS consists of $N_\text{RF}$ sub-metasurface, each of which consists of $L$ reflecting elements with the size of  $12 \times 12 $ ${\text{mm}}^2$.
The impact of the PMS is reflected in the phase shift beam ${\bf c}$. Unless specified, the optimal phase shift beam given in  Algorithm \ref{A1} is adopted.
In addition, we assume $K$ users are uniformly distributed in a disk with the radius $R=200$m. For each analytical result, 1000 random realizations of large-scale fading profiles are generated. For numerical results, they are obtained by averaging over 1000 independent small-scale fading parameters for each realization of large-scale channels.

Fig. \ref{f7} illustrates the performance of the proposed beam training scheme, where the normalized equivalent channel strength (NECS) (normalized by the ideal equivalent channel strength ) is defined as $
\frac{E\left\{   {\left\| { \bf{C H}} \right\|}^2  \right\}}{E\left\{   {\left\|    {\bf{C}}_{\text{opt}}  \bf{ H}  \right\|}^2  \right\} }$, with ${\bf{C}}_{\text{opt}} $ being the ideal phase shift matrix given by Proposition \ref{p5}. For comparison, the performance of the exhaustive scheme  and the random
selection scheme is also presented. As expected, the proposed beam training scheme significantly outperforms the random selection scheme over the entire SNR regime. Moreover, the performance of the proposed beam training scheme is close to that of the exhaustive scheme, regardless of the available number of phase shifts.

%For instance, when the received SNR is $-75$dBm and $M_{\text{ph}}=2$, the NECS of the proposed scheme is over twice that of the random selection scheme, and about $92 \%$ of the exhaustive scheme. Moreover, as the SNR increases, the NECSs of  both the proposed scheme  and  exhaustive scheme gradually saturate due to the limited number of phase shifts, while the NECS of the random selection scheme keeps constant. Moreover, when the phase shift number grows to 4, the NECSs of  both the proposed scheme  and  exhaustive become larger, because   high-accuracy beams can be obtained by increasing the number of phase shifts.

\begin{figure}[!ht]
  \centering
  \includegraphics[width=3in]{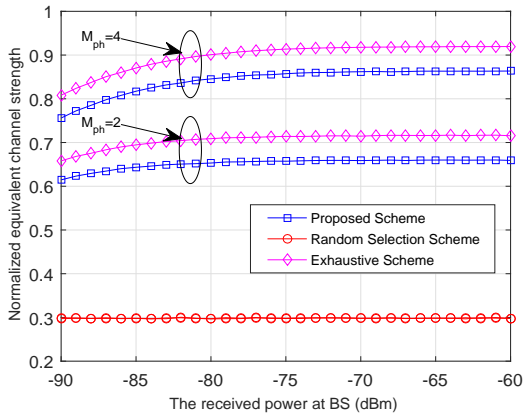}
%  %\hspace{1in}
%  \subfigure[]{\label{fig:1b}\includegraphics[scale=0.6]{system_model_version222.eps}}
   \caption{  The  performance of the proposed beam training scheme with $L=4, N_{\text{RF}}=1, K=4, \beta=0.01 $.}
  \label{f7}
\end{figure}

 Fig. \ref{f1}  shows the multicast rate with different number of RF chains, where the analytical results are
generated according to Theorem \ref{T1}. As can be readily observed, the numerical results match exactly with the analytical results, thereby validating the correctness of the analytical expressions. Moreover, the multicast rate saturates in the high SNR regime due to imperfect channel estimation. In addition, we can see that the multicast rate improves as the number of RF chains increases. The reason is that a large number of  RF chains leads to higher diversity gains.
 \begin{figure}[!ht]
  \centering
  \includegraphics[width=3in]{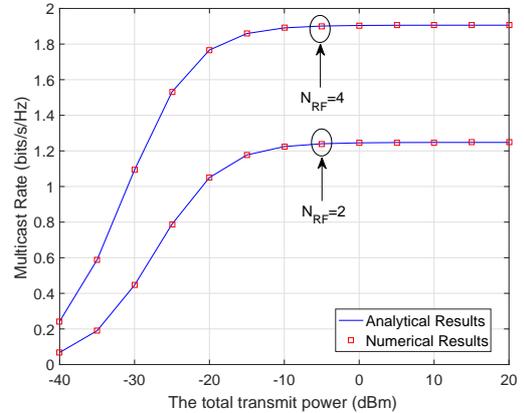}
%  %\hspace{1in}
%  \subfigure[]{\label{fig:1b}\includegraphics[scale=0.6]{system_model_version222.eps}}
   \caption{Multicast rate of PMS-based multicast systems with $L=8, K=8, M=2, \beta=0.01,\rho_p=-20 \text{dBm}$.}  
  \label{f1}
\end{figure}

 Fig. \ref{f2}  presents the multicast rate with different numbers of reflecting elements (per sub-PMS) and reflection coefficients, where the
``Approximate Results'' curve is generated according to
 Theorem \ref{t2}. As expected, the approximations well match the numerical results, especially with a larger $L$. Moreover, we can see that increasing $L$ can
significantly improve the multicast rate performance, because of the enhanced equivalent channel. Also, the multicast rate improves as the reflection coefficient $\beta$ becomes larger due to a less power loss caused by PMS reflection.
\begin{figure}[!ht]
  \centering
  \includegraphics[width=3in]{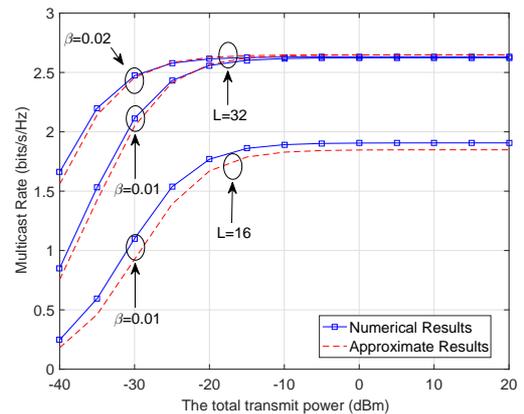}
%  %\hspace{1in}
%  \subfigure[]{\label{fig:1b}\includegraphics[scale=0.6]{system_model_version222.eps}}
   \caption{   Multicast rate of PMS-based multicast systems with $N_{\text{RF}}=4, K=8, M=2, \rho_p=-20 \text{dBm}$.}
  \label{f2}
\end{figure}

Fig. \ref{f3} shows the impact of the number of RF chains on the multicast rate with different pilot power, where the curve associated with ``Approximate Results'' is plotted according
to Theorem \ref{t3}. As the number of RF chains becomes larger, the gap between
the ``Approximate Results'' curve and the ``Numerical Results''
curve becomes smaller, which verifies our analytical results
in Theorem \ref{t3}. Moreover, we can see that as the number of RF chains becomes larger, the multicast rate keeps increasing without a ceiling, indicating that a large number of RF chains would significantly improve the multicast rate. Also, the multicast rate increases with the pilot power, due to more accurate channel estimation.

\begin{figure}[!ht]
  \centering
  \includegraphics[width=3in]{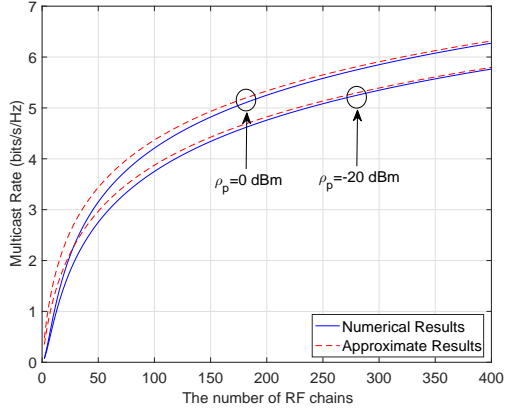}
%  %\hspace{1in}
%  \subfigure[]{\label{fig:1b}\includegraphics[scale=0.6]{system_model_version222.eps}}
   \caption{The impact of the number of RF chains with $L=2, K=8, M=2,\beta=0.01, \rho=-10 \text{dBm}$.}
  \label{f3}
\end{figure}

Fig. \ref{f4} illustrates the impact of the number of reflecting elements (per sub-PMS) on
the multicast rate, where we generate the ``Approximate Results'' curve according to Theorem \ref{t4}. As can be readily observed, the ``Approximate Results'' curve matches the ``Numerical Results'' curve well, thereby validating the correctness of Theorem \ref{t4}. Moreover, we can see that as $L$ becomes larger, the multicast rate keeps growing without a ceiling, which implies that increasing the number of reflecting elements (per sub-PMS) can always improve the multicast rate. Also, it can be observed that with the increase of   phase shift number, the multicast rate becomes larger, due to  more accurate beam training.
\begin{figure}[!ht]
  \centering
  \includegraphics[width=3in]{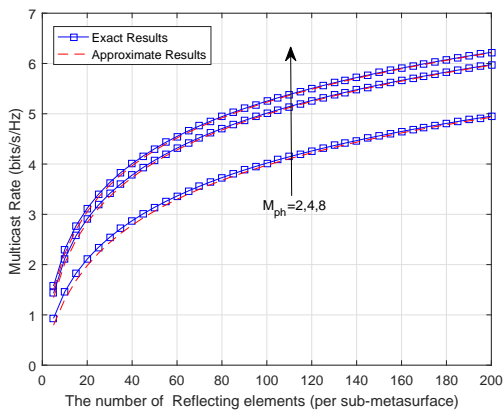}
%  %\hspace{1in}
%  \subfigure[]{\label{fig:1b}\includegraphics[scale=0.6]{system_model_version222.eps}}
   \caption{  The impact of the number of reflecting elements (per sub-PMS) with $N_{\text{RF}}=4, K=8,M=2, \beta=0.01, \rho_p=-20\text{dBm}, \rho=-10 \text{dBm}$.}
  \label{f4}
\end{figure}

Fig. \ref{f5} illustrates the impact of phase shift number on the multicast rate, where the ``Limit'' curve is plotted according to Proposition \ref{p3}. As the phase shift number becomes larger, the multicast rate gradually approaches the limit given by Proposition \ref{p3}, which verifies our analytical results. Moreover, a higher multicast rate limit can be achieved by  increasing the number of reflecting elements. This is because with massive phase shifts, the multicast rate is mainly dominated by the number of reflecting elements.
Besides, it can be seen that the  multicast rate  achieved by only few phase shifts  is comparable  to that with massive phase shifts. For instance, when $L=100$,  the multicast rate with 4 phase shifts  is about $94\%$  of that with 20 phase shifts.
\begin{figure}[!ht]
  \centering
  \includegraphics[width=3in]{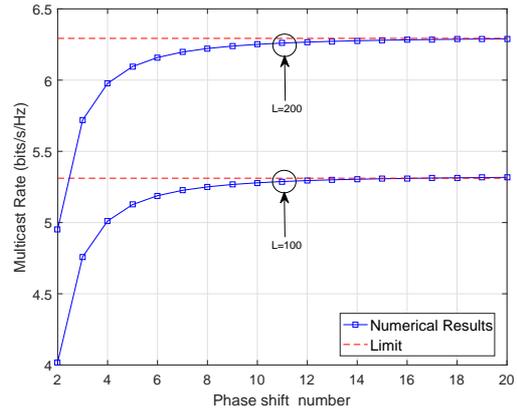}
%  %\hspace{1in}
%  \subfigure[]{\label{fig:1b}\includegraphics[scale=0.6]{system_model_version222.eps}}
   \caption{   The impact of phase shift number with $N_{\text{RF}}=4, K=8, \beta=0.01, \rho_p=-20\text{dBm}, \rho=-10 \text{dBm} $.}
  \label{f5}
\end{figure}

Fig. \ref{f6} depicts the impact of user number on the multicast rate, where the curve associated with ``Approximate Results'' is generated by Proposition \ref{p4}. As can be readily observed, the approximation is very tight, thereby verifying our analytical expressions. Moreover, the multicast rate is a decreasing function with respect to the number of users, which indicates that increasing the number of users would always degrade the multicast rate. The reason is that a large number of users would lead to poorly directional beams. In addition, we can see that increasing the number of reflecting elements can compensate the rate loss caused by the increase of  user number. For example, when the user number grows from 20 to 40, the muticast rate with $L=100$ drops from 3 bits/s/Hz to 2 bits/s/Hz. However, by increasing $L$ to $200$, the muticast rate can  remain unchanged at 3 bits/s/Hz.

\begin{figure}[!ht]
  \centering
  \includegraphics[width=3in]{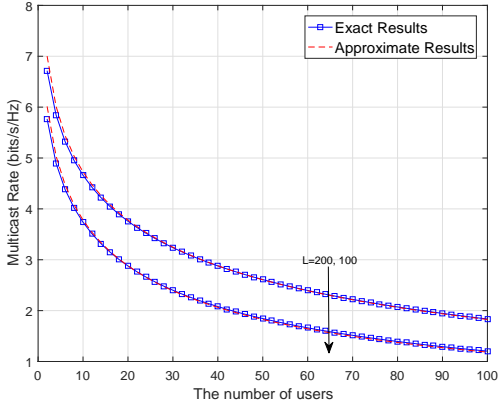}
%  %\hspace{1in}
%  \subfigure[]{\label{fig:1b}\includegraphics[scale=0.6]{system_model_version222.eps}}
   \caption{ The impact of user number with $N_{\text{RF}}=4, K=8, \rho_p=-20\text{dBm}, \rho=-10 \text{dBm}$.}
  \label{f6}
\end{figure}

Fig. \ref{transmitter} compares the proposed PMS transmitter with  a traditional multi-antenna transmitter. We can see that 
in the low-SNR regime, our proposed PMS transmitter is worse than the traditional multi-antenna transmitter due to the power loss caused by PMS reflection as well as signal propagation from the BS to the PMS. As the SNR increases, our proposed PMS transmitter becomes superior to the traditional multi-antenna transmitter. Moreover, the rate gap becomes larger as the number of reflecting elements increases, due to both the increased beam gain  and  aperture gain of the PMS.
\begin{figure}[!ht]
  \centering
  \includegraphics[width=3in]{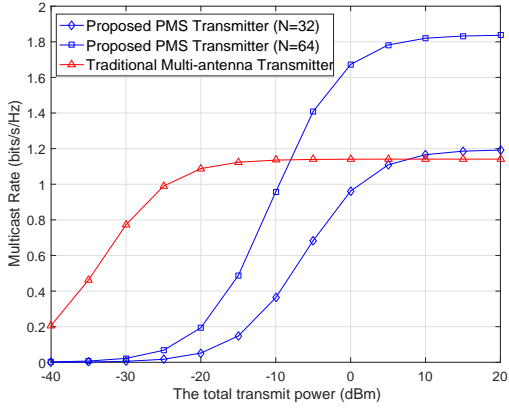}
%  %\hspace{1in}
%  \subfigure[]{\label{fig:1b}\includegraphics[scale=0.6]{system_model_version222.eps}}
   \caption{  Proposed PMS transmitter V.S. Traditional multi-antenna transmitter with $N_\text{RF}=8$, $K=8$, $M=2$, $\beta=0.01$, $\rho_p=-20 \text{dBm}$.}  
  \label{transmitter}
\end{figure}

\section{Conclusion} \label{s6}

This paper has investigated the performance of the PMS-based multicast system, taking into account of the limited resolution of phase shifts. A novel beam training algorithm has been proposed, which achieves comparable performance as the exhaustive search scheme and has much lower time overhead. Then, an exact closed-form expression for the individual user rate has been derived. Moreover, several concise asymptotical approximations for the multicast rate are presented. The analytical findings suggests that deploying a large number of RF chains or reflecting elements can greatly improve the multicast rate. Besides, as the phase shift number increases, the multicast rate gradually saturates, and the multicast rate is a decreasing function with respect to the number of users. Furthermore, with a large number of RF chains, it is better to allocate more power to users with poor channel conditions. But with large pilot power or massive reflecting elements, equal power allocation is desirable.

\begin{appendices}
\section{Proof of Proposition \ref{p1}}\label{A2}
Without loss of generality,  we focus on the $n$-th element of $\bar{{\bf h}}_k$:
$
\bar{{  h}}_{k,n}={\bf c}_{n}^{T} {\bf h}_{k,n}=\sum\limits_{l=1}^{L} c_{n}^l h_{k,n}^l.
$
Since we assume a large number of reflecting elements and limited number of RF chains, $L=\frac{N}{N_{\text{RF}}}$ is large. According to the central-limit theorem, $\bar{{  h}}_{k,n}$ approximately follows normal distribution $\mathcal{CN} \left( L u_0,L \delta_0^2\right)$, where $u_0$ and $\delta_0^2$ are the mean and  variance of $c_{n}^l h_{k,n}^l$, respectively.
In the following, we try to derive $u_0$ and $\delta_0^2$.

(1) Compute $u_0$

Denote the phase error resulted from the finite phase  shift number $M_{\text{ph}}$ by $\Delta \theta$, and we have
\begin{align}\label{E6}
u_0=\mathbb{E} \left\{ c_{n}^l h_{k,n}^l \right\} \overset{(a)}{=} \mathbb{E} \left\{ e^{j \Delta \theta} \right\}  \mathbb{E} \left\{ c_{\text{opt},n}^l h_{k,n}^l \right\},
\end{align}
where (a) is obtained according to $c_{n}^l=e^{j \Delta \theta} c_{\text{opt},n}^l$.

 We start with the computation of $\mathbb{E} \left\{ e^{j \Delta \theta} \right\}$:
 \begin{align}\label{E7}
 \mathbb{E} \{ e^{j \Delta \theta} \} \overset{(a)}{=} \int_{-\frac{\pi}{M_{\text{ph}}}}^{\frac{\pi}{M_{\text{ph}}}} e^{j \Delta \theta} \frac{M_{\text{ph}}}{2 \pi} d\Delta \theta= \frac{M_{\text{ph}}}{\pi} \sin(\frac{\pi}{M_{\text{ph}}} ),
 \end{align}
 where (a) follows the fact that $\Delta \theta \in (-\frac{\pi}{M_{\text{ph}}},\frac{\pi}{M_{\text{ph}}})$ is uniformly distributed.

 Then using the result given by Proposition \ref{p5},  we can express $\mathbb{E} \{ c_{\text{opt},n}^l h_{k,n}^l \}$ as
$
\mathbb{E} \{ c_{\text{opt},n}^l h_{k,n}^l \}= \beta  \mathbb{E} \{  \frac{h_{\text{sum},n}^{l*}}{| {  h}_{\text{sum},n}^l |} h_{k,n}^l \}.
$
Recall that ${h}_{\text{sum},n}^l=\sum\limits_{k=1}^{K} { h}_{k,n}^l$, and the following equation holds:
$
\sum\limits_{k=1}^{K} \mathbb{E}  \{ c_{\text{opt},n}^l h_{k,n}^l \}
=\beta \mathbb{E}\{  | {h}_{\text{sum},n}^l |  \}.
$
By noticing that $\left| {h}_{\text{sum},n}^l \right| $ follows Rayleigh distribution and the variance of ${h}_{\text{sum},n}^l$ is $K$, we have
$
\sum\limits_{k=1}^{K} \mathbb{E}  \{ c_{\text{opt},n}^l h_{k,n}^l \}= \frac{\beta \sqrt{\pi K}}{2}.
$
Since $\mathbb{E} \{ c_{\text{opt},n}^l h_{k_1,n}^l  \}=\mathbb{E}  \{ c_{\text{opt},n}^l h_{k_2,n}^l \}$ holds for any $(k_1,k_2)$, we have
\begin{align} \label{E8}
\mathbb{E} \left\{ c_{\text{opt},n}^l h_{k,n}^l \right\}=\frac{\beta}{2} \sqrt{\frac{\pi}{K}}.
\end{align}

Substituting (\ref{E7}) and (\ref{E8}) into (\ref{E6}), we obtain
\begin{align}
u_0=\frac{\beta}{2} \sqrt{\frac{\pi}{K}} \frac{M_{\text{ph}}}{\pi} \sin\left(\frac{\pi}{M_{\text{ph}}}\right).
\end{align}

(2) Compute $\delta_0^2$

Recall $\left|c_{\text{opt},n}^l \right|=\beta$, and we have
\begin{align}
\mathbb{E} \left\{{ \left|c_{\text{opt},n}^l h_{k,n}^l\right|}^2 \right\}
=\beta^2 \mathbb{E} \left\{{ \left|  h_{k,n}^l\right|}^2 \right\}=\beta^2 ,
\end{align}

based on which, we obtain
\begin{align}
\delta_0^2&= \mathbb{E} \left\{{ \left|c_{\text{opt},n}^l h_{k,n}^l\right|}^2 \right\}\!-\!u_0^2\\
&
=\beta^2 \left\{1\!-\!\frac{\pi}{4 K} { \left\{\frac{M_{\text{ph}}}{\pi} \sin\left(\frac{\pi}{M_{\text{ph}}}\right)\right\} }^2  \right\}.\nonumber
\end{align}

To this end, by noticing that $u=L u_0$ and $\delta^2=L \delta_0^2$, we complete our proof.

\section{Proof of Proposition \ref{p2} } \label{A3}
According to the property of MMSE, we have
\begin{align} \label{E41}
{\hat{\bar{\bf g}}}_{k}&=\mathbb{E} \left\{ \bar{\bf g}_{p,k} \right\} \\
&+\! \text{cov}\left({\bf y}_{p,k},{\bar{\bf g}}_k \right) {\left\{\text{cov}\left({\bf y}_{p,k},{\bf y}_{p,k}\right)\right\}}^{-1} \left\{ {\bf y}_{p,k}\! -\! E \left\{{\bf y}_{p,k} \right\} \right\},\nonumber
\end{align}
where 
\begin{align}
\mathbb{E} \left\{ \bar{\bf g}_{p,k} \right\}=  \sqrt{\alpha_k}u {\bf 1}_{N_{\text{RF}} \times 1},\  \mathbb{E} \left\{  {\bf y}_{p,k} \right\}=  \sqrt{\tau_p \rho_p \alpha_k} u {\bf 1}_{N_{\text{RF}} \times 1}, \nonumber
\end{align}
and ${\bf 1}_{N_{\text{RF}} \times 1}$ denotes an $N_{\text{RF}} \times 1$ vector whose elements are 1.

We first compute the covariance matrix of ${\bf y}_{p,k}$ and ${\bar{\bf g}}_k$, $ \text{cov}\left({\bf y}_{p,k},{\bar{\bf g}}_k \right)$. Using (\ref{E5}) and (\ref{E6}), we have
\begin{align} \label{E42}
 \text{cov}\left({\bf y}_{p,k},{\bar{\bf g}}_k \right)&=
 \mathbb{E}\left\{\left\{ \bar{\bf g}_{p,k}-\mathbb{E}\left\{ \bar{\bf g}_{p,k} \right\} \right\} {\left\{ {\bf y}_{p,k}-\mathbb{E}\left\{ {\bf y}_{p,k} \right\} \right\}}^H\right\} \nonumber \\
 & = \mathbb{E}\left\{    \sqrt{\alpha_k} \delta \bar{\bf g}_k
 { \left(\sqrt{\tau_p \rho_p \alpha_k}\delta {\bar{\bf g}}_k +{\bf n}_{p,k}\right)}^H
   \right\} \nonumber \\
   &= \alpha_k \sqrt{\tau_p \rho_p} \delta^2   {\bf I}_{N_{\text{RF}}}.
\end{align}

Then, we calculate $\text{cov}\left({\bf y}_{p,k},{\bf y}_{p,k}\right)$ :
\begin{align}\label{E43}
\text{cov}\left({\bf y}_{p,k},{\bf y}_{p,k}\right)
 &=\mathbb{E}\left\{\left( {\bf y}_{p,k}- \mathbb{E} \left\{  {\bf y}_{p,k} \right\} \right) {  \left( {\bf y}_{p,k}- \mathbb{E} \left\{  {\bf y}_{p,k} \right\}   \right)}^H\right\} \nonumber\\
 &=\left(1+  \tau_p \rho_p \alpha_k \delta^2 \right) {\bf I}_{N_{\text{RF}}}.
\end{align}

Substituting (\ref{E42}) and (\ref{E43}) into (\ref{E41}), we obtain
$
\hat{\bar{ {\bf g}}}_{k}= \sqrt{\alpha_k} u{ \bf  1}_{N_{\text{RF}}}
+\frac{ \alpha_k \delta^2 \sqrt{\tau_p \rho_p} }{ \tau_p \rho_p \alpha_k \delta^2+1}
\left\{ {\bf y}_{p,k}- \sqrt{  \alpha_k \tau_p \rho_p   }u { \bf  1}_{N_{\text{RF}}} \right\},
$
based on which, the covariance matrix  of $\hat{\bar{ {\bf g}}}_{k}$ is given by
\begin{align}
\text{cov}\left(\hat{\bar{{\bf g}}}_{p,k},\hat{\bar{{\bf g}}}_{p,k}\right)
 &\!=\!\mathbb{E}\left\{\!\left( \hat{\bar{{\bf g}}}_{p,k}\!-\!\mathbb{E} \left\{  \hat{\bar{{\bf g}}}_{p,k} \right\} \right) {  \left( \hat{\bar{{\bf g}}}_{p,k}\!-\!\mathbb{E} \left\{  \hat{\bar{{\bf g}}}_{p,k} \!\right\}   \right)}^H\right\} \nonumber \\
 &=\frac{\tau_p \rho_p \alpha_k^2 \delta^4}{1+ \tau_p \rho_p \alpha_k \delta^2}{\bf I}_{N_{\text{RF}}}.
\end{align}

By noticing that $\text{cov}\left({\bar{{\bf g}}}_{p,k},{\bar{{\bf g}}}_{p,k}\right)=\text{cov}\left(\hat{\bar{{\bf g}}}_{p,k},\hat{\bar{{\bf g}}}_{p,k}\right)+\text{cov}\left( { \bf e }_{k}, {\bf e }_{k}\right)$, we have
\begin{align}
\text{cov}\left({\bf e}_{k},{\bf e}_{k}\right)
=  \frac{\alpha_k \delta^2  }{ 1 +\tau_p \rho_p \alpha_k  \delta^2}{\bf I}_{N_{\text{RF}}}.
\end{align}

\section{Proof of Theorem \ref{T1}} \label{A4}
In the following, we will calculate $A_k$ and $B_k$ respectively.

1) Calculate $A_k$
\begin{align}
{A}_{k}
&= \sqrt{\rho} \mathbb{E} \left\{  \left( {\hat{\bar{\bf g}}}_k^T+{\bf e}_k^T\right) \sum\limits_{i=1}^{K} \sqrt{\frac{\eta_i}{u_{p,i}^{2}+\delta_{p,i}^{2}}}  {\hat{\bar{\bf g}}}_i^*  \right\} \\
&\overset{(a)}{=} \sqrt{\rho} \mathbb{E} \left\{    {\hat{\bar{\bf g}}}_k^T  \sum\limits_{i=1}^{K} \sqrt{\frac{\eta_i}{u_{p,i}^{2}+\delta_{p,i}^{2}}}  {\hat{\bar{\bf g}}}_i^*  \right\}\nonumber \\
&=\!\sqrt{\rho} N_{\text{RF}} \left\{\!\sum\limits_{i\!=\!1}^{K} \sqrt{\frac{\eta_i}{u_{p,i}^{2}\!+\!\delta_{p,i}^{2}}} u_{p,k} u_{p,i}\!+\!\sqrt{\frac{\eta_k}{u_{p,k}^{2}\!+\!\delta_{p,k}^{2}}} \delta_{p,k}^{2} \!\right\}. \nonumber
\end{align}
where (a) follows the fact $\mathbb{E}\left\{{\bf e}_k^T \right\}={\bf 0}$.

2) Calculate $B_k$

We first compute
\begin{align}
\mathbb{E}\left\{\!{\left|\! \left(\! {\hat{\bar{\bf g}}}_k^T\!+\!{\bf e}_k^T \!\right)\! \sum\limits_{i=1}^{K}  \sqrt{\frac{\eta_i}{u_{p,i}^{2}\!+\!\delta_{p,i}^{2}}}  {\hat{\bar{\bf g}}}_i^* \right|}^2 \right\}
 =B_{k}^{(1)}\!+\!B_{k}^{(2)},
\end{align}
where
\begin{align}
&B_{k}^{(1)}=\mathbb{E}\left\{{\left|  {\hat{\bar{\bf g}}}_k^T \sum\limits_{i=1}^{K}  \sqrt{\frac{\eta_i}{u_{p,i}^{2}+\delta_{p,i}^{2}}}  {\hat{\bar{\bf g}}}_i^* \right|}^2 \right\},\\
&B_{k}^{(2)}=\mathbb{E}\left\{{\left|  {\bf e}_k^T \sum\limits_{i=1}^{K}  \sqrt{\frac{\eta_i}{u_{p,i}^{2}+\delta_{p,i}^{2}}}  {\hat{\bar{\bf g}}}_i^* \right|}^2 \right\}.
\end{align}

We start with the calculation of the first term
:
\begin{align}
&B_{k}^{(1)}= \\
&\sum\limits_{m=1}^{K} \sum\limits_{n=1}^{K} \sqrt{\frac{\eta_m \eta_n}{ \left(u_{p,m}^{2}\!+\!\delta_{p,m}^{2}\right)\left(u_{p,n}^{2}\!+\!\delta_{p,n}^{2}\right) }}
\mathbb{E}\left\{ {\hat{\bar{\bf g}}}_k^T   {\hat{\bar{\bf g}}}_m^*  {\hat{\bar{\bf g}}}_n^T   {\hat{\bar{\bf g}}}_k^*\right\}.\nonumber
\end{align}

Let us focus on the evaluation of  $\mathbb{E}\left\{ {\hat{\bar{\bf g}}}_k^T   {\hat{\bar{\bf g}}}_m^*  {\hat{\bar{\bf g}}}_n^T   {\hat{\bar{\bf g}}}_k^*\right\}$.

a) for $m\ne n \ne k$, we have
\begin{align}
& \mathbb{E}\left\{ {\hat{\bar{\bf g}}}_k^T   {\hat{\bar{\bf g}}}_m^*  {\hat{\bar{\bf g}}}_n^T   {\hat{\bar{\bf g}}}_k^*\right\}
=\text{tr}\left(  \mathbb{E}\left\{ {\hat{\bar{\bf g}}}_m^*  \right\}  \mathbb{E}\left\{ {\hat{\bar{\bf g}}}_n^T  \right\}  \mathbb{E}\left\{ {\hat{\bar{\bf g}}}_k^* {\hat{\bar{\bf g}}}_k^T  \right\} \right)\\
&=N_{\text{RF}}^2 u_{p,m} u_{p,n}  u_{p,k}^{2}+ N_{\text{RF}}  u_{p,m} u_{p,n}\delta_{p,k}^{2}  .\nonumber
\end{align}

b) for $m\ne n = k$, we have
\begin{align}
& \mathbb{E}\left\{ {\hat{\bar{\bf g}}}_k^T   {\hat{\bar{\bf g}}}_m^*  {\hat{\bar{\bf g}}}_n^T   {\hat{\bar{\bf g}}}_k^*\right\}
=\text{tr}\left(  \mathbb{E}\left\{ {\hat{\bar{\bf g}}}_m^*  \right\}    \mathbb{E}\left\{  {\hat{\bar{\bf g}}}_k^T {\hat{\bar{\bf g}}}_k^* {\hat{\bar{\bf g}}}_k^T  \right\} \right).
\end{align}

Decomposing ${\hat{\bar{\bf g}}}_k$ into ${\bf u}_{p,k}=u_{p,k}{\bf 1 }_{N_{\text{RF}}}^T$ and ${\hat{\bar{\bf g}}}_{\text{v},k}\sim \mathcal{CN} ({\bf 0}, \delta_{p,k}^{2} {\bf I}_{N_{\text{RF}}} )$, we have
\begin{align}
& \mathbb{E}\left\{ {\hat{\bar{\bf g}}}_k^T   {\hat{\bar{\bf g}}}_m^*  {\hat{\bar{\bf g}}}_n^T   {\hat{\bar{\bf g}}}_k^*\right\} \\
&=\text{tr}\left(  {\bf u}_{p,m}^*      \mathbb{E}\left\{
 {\bf u}_{p,k}^T {\bf u}_{p,k} {\bf u}_{p,k}^T\!+\!{\hat{\bar{\bf g}}}_{\text{v},k}^T {\hat{\bar{\bf g}}}_{\text{v},k}^* {\bf u}_{p,k}^T\!+\!_{p,k}^T {\hat{\bar{\bf g}}}_{\text{v},k}^* {\hat{\bar{\bf g}}}_{\text{v},k}^T   \right\} \right) \nonumber \\
 &=N_{\text{RF}}^{2} u_{p,m} u_{p,k}^{3}+N_{\text{RF}}^{2} u_{p,m} u_{p,k} \delta_{p,k}^{2}+N_{\text{RF}}  u_{p,m} u_{p,k} \delta_{p,k}^{2}.\nonumber
\end{align}

c) for $n\ne m = k$, we have
\begin{align}
& \mathbb{E}\left\{ {\hat{\bar{\bf g}}}_k^T   {\hat{\bar{\bf g}}}_m^*  {\hat{\bar{\bf g}}}_n^T   {\hat{\bar{\bf g}}}_k^*\right\} 
=\mathbb{E}\left\{ {\hat{\bar{\bf g}}}_n^T   {\hat{\bar{\bf g}}}_k^*  {\hat{\bar{\bf g}}}_k^T   {\hat{\bar{\bf g}}}_k^*  \right\} \\ 
&=\text{tr}\left(  \mathbb{E}\left\{ {\hat{\bar{\bf g}}}_n^T  \right\}    \mathbb{E}\left\{  {\hat{\bar{\bf g}}}_k^* {\hat{\bar{\bf g}}}_k^T {\hat{\bar{\bf g}}}_k^*  \right\} \right) \nonumber\\
&=\text{tr}{\left(     \mathbb{E}\left\{
 {\bf u}_{p,k}^T {\bf u}_{p,k} {\bf u}_{p,k}^T\!+\!{\hat{\bar{\bf g}}}_{\text{v},k}^T {\hat{\bar{\bf g}}}_{\text{v},k}^* {\bf u}_{p,k}^T\!+\!{\bf u}_{p,k}^T {\hat{\bar{\bf g}}}_{\text{v},k}^* {\hat{\bar{\bf g}}}_{\text{v},k}^T   \right\} {\bf u}_{p,n}^*    \right)}^H \nonumber\\
 &=N_{\text{RF}}^{2} u_{p,n} u_{p,k}^{3}+N_{\text{RF}}^{2} u_{p,n} u_{p,k} \delta_{p,k}^{2}+N_{\text{RF}}  u_{p,n} u_{p,k} \delta_{p,k}^{2}.\nonumber
\end{align}

 d) for $m=n \ne k$, we have
\begin{align}
& \mathbb{E}\left\{ {\hat{\bar{\bf g}}}_k^T   {\hat{\bar{\bf g}}}_m^*  {\hat{\bar{\bf g}}}_n^T   {\hat{\bar{\bf g}}}_k^*\right\}
=\text{tr}\left(  \mathbb{E}\left\{ {\hat{\bar{\bf g}}}_m^* {\hat{\bar{\bf g}}}_m^T  \right\}     \mathbb{E}\left\{ {\hat{\bar{\bf g}}}_k^* {\hat{\bar{\bf g}}}_k^T  \right\} \right)\\
&=\text{tr}\left( \left\{{\bf u}_{p,m}^*{\bf u}_{p,m}^T+{\delta}_{p,m}^{2} {\bf I}_{N_{\text{RF}}} \right\} \left\{ {\bf u}_{p,k}^*{\bf u}_{p,k}^T+{\delta}_{p,k}^{2} {\bf I}_{N_{\text{RF}}}  \right\} \right) \nonumber \\
&=\!u_{p,m}^{2} u_{p,k}^2 N_{\text{RF}}^2 \!+\!u_{p,m}^{2} \delta_{p,k}^2 N_{\text{RF}}
\!+\! \delta_{p,m}^{2} u_{p,k}^2 N_{\text{RF}} \!+\!\delta_{p,m}^{2} \delta_{p,k}^2 N_{\text{RF}}.\nonumber
\end{align}

e) for $m=n = k$, we have
\begin{align}
& \mathbb{E}\left\{ {\hat{\bar{\bf g}}}_k^T   {\hat{\bar{\bf g}}}_m^*  {\hat{\bar{\bf g}}}_n^T   {\hat{\bar{\bf g}}}_k^*\right\}
= \mathbb{E}\left\{ {\hat{\bar{\bf g}}}_k^T   {\hat{\bar{\bf g}}}_k^*  {\hat{\bar{\bf g}}}_k^T   {\hat{\bar{\bf g}}}_k^*\right\} \\
&= \mathbb{E}\left\{ {\left|  {\bf u}_{p,k}^T {\bf u}_{p,k}^* +{\hat{\bar{\bf g}}}_{\text{v},k}^T{\bf u}_{p,k}^* +{\bf u}_{p,k}^T{\hat{\bar{\bf g}}}_{\text{v},k}^* +{\hat{\bar{\bf g}}}_{\text{v},k}^T{\hat{\bar{\bf g}}}_{\text{v},k}^* \right|}^2    \right\} \nonumber \\
&=\mathbb{E}\left\{  {\left|  {\bf u}_{p,k}^T {\bf u}_{p,k}^* \right|}^2
 + {\left| {\hat{\bar{\bf g}}}_{\text{v},k}^T{\bf u}_{p,k}^* \right|}^2
+{\left| {\bf u}_{p,k}^T {\hat{\bar{\bf g}}}_{\text{v},k}^* \right|}^2 \right\} \nonumber\\
&\ \ \ \  +\mathbb{E}\left\{  {\left| {\hat{\bar{\bf g}}}_{\text{v},k}^T {\hat{\bar{\bf g}}}_{\text{v},k}^*  \right|}^2
+2 {\bf u}_{p,k}^T {\bf u}_{p,k}^* {\hat{\bar{\bf g}}}_{\text{v},k}^T {\hat{\bar{\bf g}}}_{\text{v},k}^*
\right\} \nonumber \\
&=N_{\text{RF}}^2  u_{p,k}^4 \!+\! 2N_{\text{RF}}  u_{p,k}^2  \delta_{p,k}^2
\!+\!\left( N_{\text{RF}}^2\! +\!N_{\text{RF}} \right) \delta_{p,k}^4\!+\!2 N_{\text{RF}}  u_{p,k}^2  \delta_{p,k}^2  \nonumber\\
&=N_{\text{RF}}^2  u_{p,k}^4+ 4N_{\text{RF}}  u_{p,k}^2  \delta_{p,k}^2
+\left( N_{\text{RF}}^2 +N_{\text{RF}} \right) \delta_{p,k}^4. \nonumber
\end{align}

Combining a) ,b) , c), d) and e) together, we obtain
\begin{align}
B_{k}^{(1)}&=N_{\text{RF}}^2 { \left\{\sum\limits_{i=1}^{K} \sqrt{\frac{\eta_i}{u_{p,i}^{2}\!+\!\delta_{p,i}^{2}}} u_{p,k} u_{p,i}\!+\!\sqrt{\frac{\eta_k}{u_{p,k}^{2}\!+\!\delta_{p,k}^{2}}} \delta_{p,k}^{2} \right\}}^2 \nonumber \\
&+N_{\text{RF}} \left(u_{p,k}^2+\delta_{p,k}^2 \right)  \sum\limits_{i=1}^{K}
 \frac{\eta_i}{u_{p,i}^{2}+\delta_{p,i}^{2}} \delta_{p,i}^2 \nonumber\\
 &+N_{\text{RF}}\delta_{p,k}^{2} {\left\{ \sum\limits_{i=1}^{K} \sqrt{\frac{\eta_i}{u_{p,i}^{2}+\delta_{p,i}^{2}}} u_{p,i}  \right\}}^2.
\end{align}

Then, we calculate $B_{k}^{(2)}$:
\begin{align}
&B_{k}^{(2)}=\mathbb{E}\left\{{\left|  {\bf e}_k^T \sum\limits_{i=1}^{K}  \sqrt{\frac{\eta_i}{u_{p,i}^{2}+\delta_{p,i}^{2}}}  {\hat{\bar{\bf g}}}_i^* \right|}^2 \right\} \\
&=  \sum\limits_{i=1}^{K}  \sum\limits_{j=1}^{K}  \sqrt{\frac{\eta_i \eta_j}{ \left(u_{p,i}^{2}+\delta_{p,i}^{2}\right)\left(u_{p,j}^{2}+\delta_{p,j}^{2}\right) }}
\mathbb{E}\left\{ {\bf e}_{k}^{T} {\hat{\bar{\bf g}}}_i^* {\hat{\bar{\bf g}}}_j^T {\bf e}_{k}^{*} \right\}
 \nonumber\\
% &= \sum\limits_{i=1}^{K}  \sum\limits_{j \ne i}^{K}  \sqrt{\frac{\eta_i \eta_j}{ \left(u_{p,i}^{2}+\delta_{p,i}^{2}\right)\left(u_{p,j}^{2}+\delta_{p,j}^{2}\right) }}
% \text{tr} \left(  E\left\{ {\hat{\bar{\bf g}}}_i^* {\hat{\bar{\bf g}}}_j^T  \right\}
%  E\left\{ {\bf e}_{k}^{T} {\bf e}_{k}^{*}  \right\} \right) \nonumber\\
%& +\sum\limits_{i=1}^{K} \frac{\eta_i}{u_{p,i}^{2}+\delta_{p,i}^{2}}  \text{tr} \left(  E\left\{ {\hat{\bar{\bf g}}}_i^* {\hat{\bar{\bf g}}}_i^T  \right\}
%  E\left\{ {\bf e}_{k}^{T} {\bf e}_{k}^{*}  \right\} \right)
% \nonumber\\
 &=\!N_{\text{RF}} \delta_{e,k}^2 {\left(\! \sum\limits_{i=1}^{K} \!\sqrt{\!\frac{\eta_i  u_{p,i}^2}{u_{p,i}^{2}\!+\!\delta_{p,i}^{2}}} \!\right)\!}^2
 \!+\!N_{\text{RF}} \delta_{e,k}^2 \sum\limits_{i=1}^{K}  \frac{\eta_i}{u_{p,i}^{2}\!+\!\delta_{p,i}^{2}} \delta_{p,i}^2. \nonumber
\end{align}

Noticing that $B_k=\rho \left( B_k^{(1)} + B_k^{(2)} \right)-{\left|A_k \right|}^2$, we obtain
\begin{align}
B_k&= N_{\text{RF}} \rho \delta_{p,k}^{2} {\left\{ \sum\limits_{i=1}^{K} \sqrt{\frac{\eta_i}{u_{p,i}^{2}+\delta_{p,i}^{2}}} u_{p,i}  \right\}}^2 \\
&+\!N_{\text{RF}} \rho \left(u_{p,k}^2\!+\!\delta_{p,k}^2 \right)  \sum\limits_{i=1}^{K}
 \frac{\eta_i}{u_{p,i}^{2}\!+\!\delta_{p,i}^{2}} \delta_{p,i}^2 \!+\!N_{\text{RF}}\rho \delta_{e,k}^2\nonumber
 \\
&\! \times\!
 {\left(\! \sum\limits_{i=1}^{K} \sqrt{\frac{\eta_i}{u_{p,i}^{2}\!+\!\delta_{p,i}^{2}}} u_{p,i} \!\right)}^2\!+\!N_{\text{RF}} \rho \delta_{e,k}^2 \sum\limits_{i=1}^{K}  \frac{\eta_i}{u_{p,i}^{2}\!+\!\delta_{p,i}^{2}} \delta_{p,i}^2. \nonumber
\end{align}

Recall that $u_{p,k}^2=\alpha_k u^2$ and $\delta_{p,k}^2 +\delta_{e,k}^2 = \alpha_k \delta^2$. The above equation can be rewritten as
\begin{align}
B_k&=\alpha_k \rho \delta^2  N_{\text{RF}} {\left( \sum\limits_{i=1}^{K} \sqrt{\frac{\eta_i u_{p,i}^2}{u_{p,i}^{2}+\delta_{p,i}^{2}}} \right)}^2 \\
&+\alpha_k \rho \left(u^2+\delta^2\right) N_{\text{RF}} \sum\limits_{i=1}^{K}  \frac{\eta_i \delta_{p,i}^2} {u_{p,i}^{2}+\delta_{p,i}^{2}}.\nonumber
\end{align}

Combining 1) and 2), we obtain the desired result.

\section{Proof of Theorem \ref{t2}} \label{A5}
As  $\rho_p \rightarrow \infty $, we have
$u_{p,i}=\sqrt{\alpha_i}u$ and $  \delta_{p,i}^{2}=\alpha_i \delta^2$,
based on which, the achievable rate of the $k$-th user can be expressed as
\begin{align}
& R_k= \\
& \log_2 \left\{1+\frac{ \frac{N_{\text{RF}}^{2} \rho \alpha_k }{u^2+\delta^2}
 {\left( \sum\limits_{i=1}^{K} \sqrt{\eta_i} u^2+\sqrt{\eta_k} \delta^2  \right)}^2}
 {1+\frac{N_{\text{RF}} \rho \alpha_k }{u^2+\delta^2}u^2 \delta^2 {\left(\sum\limits_{i=1}^{K}\sqrt{\eta_i}  \right)}^2 +N_{\text{RF}} \rho \alpha_k \delta^2 \sum\limits_{i=1}^{K} \eta_i  } \right\}.\nonumber
\end{align}

Noticing that $\sum\limits_{i=1}^{K} \eta_i =1$, we have
\begin{align} \label{E15}
&R_k\!= \\
&\!\log_2 \!\left\{\!1\!+\!\frac{ \frac{N_{\text{RF}}^{2} \rho \alpha_k }{u^2\!+\!\delta^2}
 {\left(\! \sum\limits_{i=1}^{K} \sqrt{\eta_i} u^2\!+\!\sqrt{\eta_k} \delta^2  \!\right)}^2}
 {1\!+\!\frac{N_{\text{RF}} \rho \alpha_k }{u^2\!+\!\delta^2}u^2 \delta^2 {\left(\!\sum\limits_{i=1}^{K}\sqrt{\eta_i}  \!\right)}^2 \!+\!N_{\text{RF}} \rho \alpha_k \delta^2  } \!\right\}.\nonumber
\end{align}
Comparing $\sum\limits_{i=1}^{K} \sqrt{\eta_i} u^2$ with  $\sqrt{\eta_k} \delta^2$, we have
$
\frac{\sum\limits_{i=1}^{K} \sqrt{\eta_i} u^2}{\sqrt{\eta_k} \delta^2 }=\frac{L \tilde{u}_{0}}{1-\tilde{u}_{0}}
{ (1+\sum\limits_{i \ne k}^{K}\sqrt{\frac{\eta_i}{\eta_k}}  )}^2
$
, which shows that $\sum\limits_{i=1}^{K} \sqrt{\eta_i} u^2$ is much greater than $\sqrt{\eta_k} \delta^2$, due to the fact that $L$ and $K$ are large in general.
Thus, ignoring the term $\sqrt{\eta_k} \delta^2$ in (\ref{E15}), we have the following approximation:
\begin{align}  \label{E16}
&R_k \approx \\
&\log_2 \left\{1+\frac{ \frac{N_{\text{RF}}^{2} \rho \alpha_k }{u^2+\delta^2}u^4
 {\left( \sum\limits_{i=1}^{K} \sqrt{\eta_i}  \right)}^2}
 {1+\frac{N_{\text{RF}} \rho \alpha_k }{u^2+\delta^2}u^2 \delta^2 {\left(\sum\limits_{i=1}^{K}\sqrt{\eta_i}  \right)}^2 +N_{\text{RF}} \rho \alpha_k \delta^2  } \right\}.\nonumber
\end{align}

Next, we consider the maximum power control problem, which  can be formulated as
\begin{align}
  \begin{array}{ll}
    \max\limits_{\left\{ {\eta}_{i} \right\}}
& \min\limits_{i=1,...,K} {R}_{i},  \\
 \operatorname{s.t.} & \begin{array}[t]{lll}
             \sum\limits_{i=1}^{K} {\eta}_{i}  = 1, \  {\eta}_{i} \geq 0, i=1,...,K. 
           \end{array}
  \end{array}
\end{align}

According to (\ref{E16}), we can observe
\begin{align}
\min\limits_{i=1,...,K} {R}_{i}= R_{i^*}, i^*=\underset{i=1,...,K}{\arg} \min \alpha_i,
\end{align}
based on which, the above optimization problem can be rewritten as
\begin{align} \label{E17}
  \begin{array}{ll}
    \max\limits_{\left\{ {\eta}_{i} \right\}}
& \log_2 \left\{\!1\!+\!\frac{ \frac{N_{\text{RF}}^{2} \rho \alpha }{u^2\!+\!\delta^2}u^4
 {\left( \sum\limits_{i=1}^{K} \sqrt{\eta_i}  \right)}^2}
 {1\!+\!\frac{N_{\text{RF}} \rho \alpha}{u^2\!+\!\delta^2}u^2 \delta^2 {\left(\sum\limits_{i=1}^{K}\sqrt{\eta_i}  \right)}^2\! +\!N_{\text{RF}} \rho \alpha \delta^2  } \!\right\},  \\
 \operatorname{s.t.} & \begin{array}[t]{lll}
             \sum\limits_{i=1}^{K} {\eta}_{i}  = 1  
             {\eta}_{i} \geq 0, i=1,...,K,
           \end{array}
  \end{array}
\end{align}
where we define $\alpha \triangleq \underset{k=1,...,K}{\min} \alpha_k$.

Noticing that the objective function is an increasing function with respect to $\sum\limits_{i=1}^{K} \sqrt{\eta_i}$, the optimization problem  is equivalent to
\begin{align}
  \begin{array}{ll}
    \max\limits_{\left\{ {\eta}_{i} \right\}}
&\sum\limits_{i=1}^{K} \sqrt{\eta_i} ,  \\
 \operatorname{s.t.} & \begin{array}[t]{lll}
             \sum\limits_{i=1}^{K} {\eta}_{i}  = 1 ,\ 
             {\eta}_{i} \geq 0, i=1,...,K.
           \end{array}
  \end{array}
\end{align}

Denote $x_i \triangleq  \sqrt{\eta_i}$. Then we have a convex problem:
\begin{align}
  \begin{array}{ll}
    \max\limits_{\left\{ x_i \right\}}
&\sum\limits_{i=1}^{K} x_i ,  \\
 \operatorname{s.t.} & \begin{array}[t]{lll}
             \sum\limits_{i=1}^{K} x_i^2 = 1, \ 
             {x}_{i} \geq 0, i=1,...,K.
           \end{array}
  \end{array}
\end{align}

By applying KKT conditions, we can obtain the optimal power control coefficients $\eta_i=\frac{1}{K}, i=1,...,K$.
To this end, substituting the optimal coefficients into the objective function of (\ref{E17}), we complete our proof.

\section{Proof of Theorem \ref{t3}} \label{A6}
Starting from Theorem  \ref{T1}, we have
\begin{align}
&\underset{N_{\text{RF}} \to \infty } {R_k} =\\
&\log_2 \left(1+\frac{     N_{\text{RF}}  {\!\left\{\!\sum\limits_{i\!=\!1}^{K} \sqrt{\frac{\eta_i}{u_{p,i}^{2}\!+\!\delta_{p,i}^{2}}} u_{p,k} u_{p,i}\!+\!\sqrt{\frac{\eta_k}{u_{p,k}^{2}\!+\!\delta_{p,k}^{2}}} \delta_{p,k}^{2} \!\right\}}^2 }    {\alpha_k  \delta^2 {\left(\! \sum\limits_{i\!=\!1}^{K} \sqrt{\frac{\eta_i u_{p,i}^2}{u_{p,i}^{2}\!+\!\delta_{p,i}^{2}}} \!\right)}^2
\!+\!\alpha_k   \left(\!u^2\!+\!\delta^2\!\right) \sum\limits_{i\!=\!1}^{K}  \frac{\eta_i \delta_{p,i}^2}{u_{p,i}^{2}\!+\!\delta_{p,i}^{2}}} \!\right). \nonumber
\end{align}

Recall that $u_{p,k}=\sqrt{\alpha_k} u$, and the above equation can be written as
\begin{align} \label{E21}
&\underset{N_{\text{RF}} \to \infty } {R_k} =\\
&\log_2 \left(1+\frac{     N_{\text{RF}}  {\left\{\sum\limits_{i=1}^{K} \sqrt{\frac{\eta_i u_{p,i}^2}{u_{p,i}^{2}+\delta_{p,i}^{2}}} u +\sqrt{\frac{\eta_k \delta_{p,k}^{4}}{ \alpha_k \left(u_{p,k}^{2}+\delta_{p,k}^{2}\right)}} \right\}}^2 }    {  \delta^2 {\left( \sum\limits_{i=1}^{K} \sqrt{\frac{\eta_i u_{p,i}^2}{u_{p,i}^{2}+\delta_{p,i}^{2}}} \right)}^2
+   \left(u^2+\delta^2\right) \sum\limits_{i=1}^{K}  \frac{\eta_i \delta_{p,i}^2}{u_{p,i}^{2}+\delta_{p,i}^{2}}} \right).\nonumber
\end{align}

Next, we try to deal with the maximum power control problem:
\begin{align}
  \begin{array}{ll}
    \max\limits_{\left\{ {\eta}_{i} \right\}}
& \min\limits_{k=1,...,K} \underset{N_{\text{RF}} \to \infty } {R_k},  \\
 \operatorname{s.t.} & \begin{array}[t]{lll}
             \sum\limits_{i=1}^{K} {\eta}_{i}  = 1, \ 
             {\eta}_{i} \geq 0, i=1,...,K,
           \end{array}
  \end{array}
\end{align}
which is equivalent to
\begin{align}
  \begin{array}{ll}
    \max\limits_{\left\{ {\eta}_{i} \right\}}
& \min\limits_{k=1,...,K} {SINR}_{k} ,  \\
 \operatorname{s.t.} & \begin{array}[t]{lll}
             \sum\limits_{i=1}^{K} {\eta}_{i}  = 1 ,\ 
             {\eta}_{i} \geq 0, i=1,...,K,
           \end{array}
  \end{array}
\end{align}
where
\begin{align}
  {SINR}_{k}=\frac{     N_{\text{RF}}  {\left\{\sum\limits_{i=1}^{K} \sqrt{\frac{\eta_i u_{p,i}^2}{u_{p,i}^{2}+\delta_{p,i}^{2}}} u +\sqrt{\frac{\eta_k \delta_{p,k}^{4}}{ \alpha_k \left(u_{p,k}^{2}+\delta_{p,k}^{2}\right)}} \right\}}^2 }    {  \delta^2 {\left( \sum\limits_{i=1}^{K} \sqrt{\frac{\eta_i u_{p,i}^2}{u_{p,i}^{2}+\delta_{p,i}^{2}}} \right)}^2
+   \left(u^2+\delta^2\right) \sum\limits_{i=1}^{K}  \frac{\eta_i \delta_{p,i}^2}{u_{p,i}^{2}+\delta_{p,i}^{2}}}.  
\end{align}

%Let $x_i \triangleq \sqrt{\eta_i}$,  we have
%\begin{align} \label{E18}
%  \begin{array}{ll}
%    \max\limits_{\left\{ {x}_{i} \right\}}
%& \min\limits_{k=1,...,K} {SINR}_{k}=\frac{     N_{\text{RF}}  {\left\{\sum\limits_{i=1}^{K} x_i\sqrt{\frac{  u_{p,i}^2}{u_{p,i}^{2}+\delta_{p,i}^{2}}} u +x_k\sqrt{\frac{ \delta_{p,k}^{4}}{ \alpha_k \left(u_{p,k}^{2}+\delta_{p,k}^{2}\right)}} \right\}}^2 }    {  \delta^2 {\left( \sum\limits_{i=1}^{K}x_i \sqrt{\frac{ u_{p,i}^2}{u_{p,i}^{2}+\delta_{p,i}^{2}}} \right)}^2
%+   \left(u^2+\delta^2\right) \sum\limits_{i=1}^{K}  \frac{x_i^2 \delta_{p,i}^2}{u_{p,i}^{2}+\delta_{p,i}^{2}}} ,  \\
% \operatorname{s.t.} & \begin{array}[t]{lll}
%             \sum\limits_{i=1}^{K} {x}_{i}^2  = 1 \\
%             {x}_{i} \geq 0, i=1,...,K.
%           \end{array}
%  \end{array}
%\end{align}
%
%We can see that the problem (\ref{E18}) is a quasi-concave optimization problem \cite{hu2019cell}.

With optimal power control coefficients, the following equation holds:
$
{SINR}_{1}=...={SINR}_{k}=...={SINR}_{K}
$,
which can be simplified as
\begin{align} \label{E19}
 & \frac{\eta_1 \delta_{p,1}^4}{\alpha_1 \left(u_{p,1}^{2} +\delta_{p,1}^{2} \right)}=...=\frac{\eta_k \delta_{p,k}^4}{\alpha_k \left(u_{p,k}^{2} +\delta_{p,k}^{2} \right)}\\
 & =...= \frac{\eta_K \delta_{p,K}^4}{\alpha_K \left(u_{p,K}^{2} +\delta_{p,K}^{2} \right)}.\nonumber
\end{align}

Utilizing (\ref{E19}) and noticing $\sum\limits_{k=1}^{K} \eta_k=1$, we can obtain the optimal power control coefficient
\begin{align} \label{E20}
\eta_k=\frac{ \alpha_k \left(u_{p,k}^{2} +\delta_{p,k}^2\right) }{\phi \delta_{p,k}^4},\ \phi=\sum\limits_{k=1}^{K} \frac{ \alpha_k \left(u_{p,k}^{2} +\delta_{p,k}^2 \right)}{   \delta_{p,k}^4}.
\end{align}

Then, substituting (\ref{E20}) into (\ref{E21}), we have
\begin{align}\label{E22}
R=\log_2 \left(1+\frac{N_{\text{RF}} {\left(u \sum\limits_{i=1}^{K} \frac{ \sqrt{\alpha_i}  u_{p,i}}{\delta_{p,i}^{2}}  +1\right)}^2   }
{\delta^2 {\left(\sum\limits_{i=1}^{K} \frac{ \sqrt{\alpha_i}  u_{p,i}}{\delta_{p,i}^{2}} \right)}^2 +
 \left(u^2+\delta^2 \right) \sum\limits_{i=1}^{K} \frac{\alpha_i}{\delta_{p,i}^{2}} } \right).
\end{align}

Finally, substituting the results given by Proposition \ref{p2} into (\ref{E22}) yields the desired result.

\section{Proof of Theorem \ref{t4}}\label{A7}

 As $L \to \infty$, we have $\delta^2=L \beta^2 \left( 1-\tilde u_0^2 \right) \to \infty$.
Then, using the results given in Proposition \ref{p2}, we have
$
\delta_{p,i}^{2}=\alpha_i \delta^2 \frac{\tau_p \rho_p \alpha_k}{\frac{1}{\delta^2}+\tau_p \rho_p \alpha_i \delta^2 }
\approx \alpha_i \delta^2.
$

Leveraging the above equation and noticing that $u_{p,i}=\sqrt{\alpha_i} u$, we obtain
\begin{align} \label{E13}
R_k=\log_2 \left\{1+\frac{ \frac{N_{\text{RF}}^{2} \rho \alpha_k }{u^2+\delta^2}
 {\left( \sum\limits_{i=1}^{K} \sqrt{\eta_i} u^2+\sqrt{\eta_k} \delta^2  \right)}^2}
 {1+\frac{N_{\text{RF}} \rho \alpha_k }{u^2+\delta^2}u^2 \delta^2 {\left(\sum\limits_{i=1}^{K}\sqrt{\eta_i}  \right)}^2 +N_{\text{RF}} \rho \alpha_k \delta^2  } \right\}.
\end{align}

Substituting $u=L \beta \tilde u_0$ and $\delta^2=L \beta^2-L\beta \tilde u_0^2$ into (\ref{E13}) and after some manipulations, we express $R_k$ as  (\ref{EF2}) given on the top of the next page.
\newcounter{mytempeqncnt2}
\begin{figure*}[!t]
\normalsize
\setcounter{mytempeqncnt2}{\value{equation}}
\begin{align} \label{EF2}
R_k=\log_2 \left\{1+\frac{N_{\text{RF}}  \beta^2 L {\left\{
\sum\limits_{i=1}^{K} \sqrt{\eta_i} \tilde u_0^2 L+
\sqrt{\eta_k} \left( 1-\tilde u_0^2\right)\right\}}^2}
{\beta^2 \tilde u_0^2\left\{ {\left(\sum\limits_{i=1}^{K} \sqrt{\eta_i}\right)}^2 +1\right\}\left(1-\tilde u_0^2 \right)L^2+
\beta^2 {\left(1-\tilde u_0^2\right)}^2L +
\frac{\beta \tilde u_0^2+1-\tilde u_0^2}{N_{\text{RF}} \alpha_k \rho} } \right\}.
\end{align}
\vspace{-0.8cm}
\end{figure*}

Neglecting the small items that do not scale with $L^2$, the above equation can be simplified as
\begin{align} \label{E23}
R_k=\log_2 \left\{1+\frac{N_{\text{RF}} L \tilde u_0^2 {\left( \sum\limits_{i=1}^{K} \sqrt{\eta_i}\right)}^2 }{\left( 1-\tilde u_0^2 \right) \left\{ {\left( \sum\limits_{i=1}^{K} \sqrt{\eta_i}\right)}^2+1 \right\}} \right\},
\end{align}
which is an increasing function with respect to $ \sum\limits_{i=1}^{K} \sqrt{\eta_i}$.
Then, following the similar process in the proof of Theorem \ref{t2}, we can obtain the optimal maximum power control coefficients $\eta_k=\frac{1}{K}, k=1,...,K$.

To this end, substituting the optimal power control coefficients into (\ref{E23}) yields the desired result.
\end{appendices}
%\bibliographystyle{unsrt}
%\bibliography{references}

%\cite{liu2016flexible}
%\cite{xiao2017millimeter}

\bibliographystyle{IEEEtran}
\bibliography{references}{}
\end{document}